\documentclass[prodmode,acmec]{ec-acmsmall} 

\usepackage[ruled]{algorithm2e}

\usepackage[numbers]{natbib}

\usepackage{amsfonts}
\usepackage{amsfonts}
\usepackage{booktabs}
\usepackage{epsfig} 
\usepackage{subfigure}
\usepackage{slashbox}
\usepackage{amssymb}
\usepackage{amsmath}
\usepackage{multirow}
\usepackage{bbm}

\begin{document}



\title{Generalized Second Price Auction with Probabilistic Broad Match}
\author{
Wei Chen
\affil{Microsoft Research Asia, Beijing, China}
Di He
\affil{Peking University, Beijing, China}
Tie-Yan Liu
\affil{Microsoft Research Asia, Beijing, China}
Tao Qin
\affil{Microsoft Research Asia, Beijing, China}
Yixin Tao
\affil{Shanghai Jiao Tong University, Shanghai, China}
Liwei Wang
\affil{Peking University, Beijing, China}
}


\begin{abstract}
Generalized Second Price (GSP) auctions are widely used by search engines today to sell their ad slots. Most search engines have supported broad match between queries and bid keywords when executing GSP auctions, however, it has been revealed that GSP auction with standard broad-match mechanism they are currently using (denoted as \textsf{SBM-GSP}) has several theoretical drawbacks (e.g., its theoretical properties are known only for the single-slot case and full-information setting, and even in this simple setting, the corresponding worst-case social welfare can be rather bad). To address this issue, we propose a novel broad-match mechanism, which we call the \emph{Probabilistic Broad-Match} (\textsf{PBM}) mechanism. Different from \textsf{SBM} that puts together the ads bidding on all the keywords matched to a given query for the GSP auction, the GSP with \textsf{PBM} (denoted as \textsf{PBM-GSP}) randomly samples a keyword according to a predefined probability distribution and only runs the GSP auction for the ads bidding on this sampled keyword. We perform a comprehensive study on the theoretical properties of the \textsf{PBM-GSP}. Specifically, we study its social welfare in the worst equilibrium, in both full-information and Bayesian settings. The results show that \textsf{PBM-GSP} can generate larger welfare than \textsf{SBM-GSP} under mild conditions. Furthermore, we also study the revenue guarantee for \textsf{PBM-GSP} in Bayesian setting. To the best of our knowledge, this is the first work on broad-match mechanisms for GSP that goes beyond the single-slot case and the full-information setting. \end{abstract}

\category{J.4}{Computer Applications}{Social and Behavioral Sciences|Economics}

\terms{Economics, Theory}

\keywords{Auction Theory, Mechanism Design, Price of Anarchy, Sponsored Search}


\maketitle

\section{Introduction}

Online advertising has become a key revenue source for many businesses on the Internet. Sponsored search is a major type of online advertising, which displays paid advertisements (ads) along with organic search results. Generalized Second Price (GSP) auction is one of the most commonly used auction mechanisms in sponsored search, which works as follows. When a query is issued by a web user, the search engine ranks all the ads bidding on this query (or keywords related to the query) according to their bid prices, and charges the owner of a clicked ad by the minimum bid price for him/her to maintain the current rank position.\footnote{In practice, the predicted click-through rate is also used in the ranking and pricing rules. However, it can be safely absorbed into the \emph{weighted} bid prices without influencing the theoretical analysis on the GSP auctions.}

If only the ads that exactly bid on the query are included in the auction, we call the corresponding mechanism an \emph{exact-match} mechanism. The GSP auction in this specific setting has been well studied in the literature \cite{babaioff2010equilibrium,lahaie2007revenue,goel2009hybrid,varian2007position,aggarwal2006truthful,caragiannis2011efficiency,gsp}, and has been shown to have a number of nice theoretical properties: (1) It possesses an efficient (welfare-maximizing) Nash equilibrium; (2) Its social welfare in equilibrium is fairly good even in the worst case : the pure price of anarchy (PoA) is bounded by 1.282 and the Bayes-Nash PoA is bounded by 2.927; (3) In the Bayesian setting, the GSP auction paired with the Myerson reserve price generates at least a constant fraction (i.e., $1/6$) of the optimal revenue in its Bayes-Nash equilibria for MHR distribution.

Despite the fruitful and positive results, the exact-match mechanism is not sufficient when we are faced with practical requirements in commercial search engines. First, the query space is extremely large (billions of queries are issued by web users every day), so it is practically impossible for advertisers to bid on every query related to their ads. Second, even if advertisers are capable enough to bid on the huge number of related queries, the search engine might not be able to afford it due to the scalability and latency constraints. Due to these reasons, commercial search engines usually use a \emph{broad-match} mechanism to enhance the GSP auction. A broad-match mechanism requires advertisers to bid on at most $\kappa$ keywords instead of an arbitrary number of queries, and matches the keywords to queries using a query-keyword bipartite graph (in which the number of keywords is significantly smaller than the number of queries). The broad-match mechanism is friendly to advertisers since they only need to consider a relatively small number of keywords in order to reach a large number of related queries. The mechanism is also friendly to the search engine since it restricts the complexity of the bidding language and therefore that of the auction system.

Today, most search engines implement the broad-match mechanism in a straightforward manner. That is, when a query is issued, all the ads bidding on the keywords that can be matched to the query on the query-keyword bipartite graph will be put together for the GSP auction. And for every advertiser, the bids on the matched keywords will be transformed to the bid on the query using some pre-defined heuristics (e.g., the maximum bid on the matched keywords). For ease of reference, we call the broad-match mechanism described above as the \emph{Standard Broad-Match} GSP mechanism, or \textsf{SBM-GSP} for short.

Although this mechanism effectively addresses the problems with the exact-match mechanism, as far as we know, it has several theoretical drawbacks.
\begin{itemize}
\item The social welfare of the \textsf{SBM-GSP} mechanism was studied in \cite{dhangwatnotai2011multi}, for the single-slot case and full-information setting only. By using the notion of homogeneity (denoted as $c$) to measure the diversity of an advertiser's valuations over different queries that can be matched to a keyword, an almost-tight pure PoA bound was derived, whose order is $\Theta(c^2)$. Considering that $c$ is usually large in practice, it can be concluded that the social welfare of the \textsf{SBM-GSP} mechanism can be rather bad in its worst equilibrium.
\item One has not obtained a complete picture about the theoretical properties of the \textsf{SBM-GSP} mechanism: no results are available regarding the multi-slot case (which is, however, more practically important since most search engines sell multiple ad slots per query), and even for the single-slot case, the social welfare and revenue in the Bayesian setting are not clear.
\end{itemize}

Given the aforementioned limitations of the \textsf{SBM-GSP} mechanism, a natural question to ask is whether we can design a broad-match mechanism with better guarantees on its performance, in terms of both social welfare and revenue, for both single-slot and multi-slot cases, and in both full-information and Bayesian settings. This is exactly the focus of our work.

In this paper, we propose a new broad-match mechanism, which we call \emph{Probabilistic Broad-Match} mechanism. Its basic idea is as follows. For each query, our mechanism assigns a matching probability to every keyword that can be matched to this query on the query-keyword bipartite graph. When the query is issued by a user, the mechanism randomly chooses a keyword according to the matching probability distribution and runs the GSP auction only upon those ads that bid on the chosen keyword. For simplicity, we also use \textsf{PBM-GSP} to refer to the above mechanism.

We perform a comprehensive study on the social welfare in equilibrium of the \textsf{PBM-GSP} mechanism, for both single-slot and multi-slot cases, and in both full-information and Bayesian settings. We also derive a revenue bound for the \textsf{PBM-GSP} mechanism for both single-slot and multi-slot cases in the Bayesian setting. To the best of our knowledge, this is the first work on broad-match mechanisms that goes far beyond the single-slot case and the full-information setting.

\paragraph{Our Results} The contributions of our work can be summarized as follows.
\begin{itemize}
\item (Section 3) We propose a novel broad-match mechanism (i.e., the \textsf{PBM} mechanism) for multi-slot sponsored search auctions.
\item (Section 4) We analyze the social welfare in equilibrium of the \textsf{PBM-GSP} mechanism in both full-information and Bayesian settings. We define a new concept, called keyword-level expressiveness (denoted as $\beta$), which can better characterize the expressiveness of the bidding language in the \textsf{PBM-GSP} mechanism than the concept of expressiveness proposed in previous work \cite{dhangwatnotai2011multi}.
\begin{itemize}
\item (Section 4.1) We extend the concept of homogeneity $c$ defined in \cite{dhangwatnotai2011multi} to the Bayesian setting, and prove that the Bayes-Nash PoA of \textsf{PBM-GSP} is at most $\frac{ec(1+\beta)}{(e-1)\beta}$ in the multi-slot case. The bound can be further optimized to $\frac{c(1+\beta)}{\beta}$ in the single-slot case.
\item (Section 4.2) We prove that in the full-information setting, the pure PoA of \textsf{PBM-GSP} is at most $\frac{c(1+\beta)}{\beta}$ when there are multiple slots to display ads. And the bound can be improved to $\frac{c}{\beta}$ in the single-slot case (which is tight with respect to each factor). Furthermore, we show that the pure PoA bound of \textsf{PBM-GSP} is better than that of \textsf{SBM-GSP} in the same setting under mild conditions.
\end{itemize}
\item (Section 5) We analyze the revenue bound of \textsf{PBM-GSP} in the Bayesian setting. We prove that by using the Myerson reserve price to each keyword, \textsf{PBM-GSP} can achieve a revenue at least $\frac{\beta}{1 + \beta}\frac{1}{2\eta(ce)^2}$ of the optimal social welfare with MHR distribution, where $\eta$ is the maximum derivative of the virtual value function.
\end{itemize}

\section{Preliminaries}
In this section, we introduce the basics about broad-match auctions, and some preliminary concepts that will be used in our theoretical analysis.

\subsection{Broad-Match Auctions}

According to \cite{feldman2007budget,broder2009online,even2009bid,dhangwatnotai2011multi}, a broad-match mechanism can be defined on a query-keyword bipartite graph. Denote $\mathcal{Q}$ as the query space, and denote $P$ as a probability distribution over $\mathcal{Q}$, which indicates the probability that query $q$ is issued by users. Denote $\mathcal{S}$ as the keyword space. In practice, the size of $\mathcal{Q}$ is much larger than the size of $\mathcal{S}$. Denote $G =(\mathcal{Q}, \mathcal{S}, \mathcal{E})$ as a (undirected) bipartite graph between queries and keywords, in which an edge $(q,s)\in\mathcal{E}$ if and only if query $q$ can be matched to keyword $s$ (or equivalently, $s$ can be matched to $q$).  Denote $N_{G}(v)$ as the neighborhood of vertex $v\in \mathcal{Q}\cup\mathcal{S}$, i.e., for any query $q$, $N_{G}(q)=\{s: (q,s)\in\mathcal{E}\}$ represents the set of keywords that can be matched to the query, and for any keyword $s$, $N_{G}(s)=\{q: (q,s)\in\mathcal{E}\}$ represents the set of queries that can be matched to the keyword. Without loss of generality, we assume $N_{G}(s)\neq \emptyset$, for all $ s$ and $N_{G}(q)\neq \emptyset$, for all $ q$.

Assume there are $n$ advertisers and $n$ slots. Denote  $w_k$ as the click probability associated with the $k$-th ad slot\footnote{In real world, the slot number is usually bounded by a constant $K$. In this case, we can define $w_k=0,k>K$ without loss of any generality.}, which satisfies $w_{i}\geq w_j$ i.f.f $i<j$. We assume advertiser $i$ has a private valuation $v_i^q$ for query $q\in\mathcal{Q}$ if his/her ad is clicked by the users, denote $\mathbf{v}=(v_1,v_2,...,v_n)$ as the valuation profile of advertisers in which $v_i \in R_{+}^{|\mathcal{Q}|}$ is the vector indicating the $i$-th advertiser's valuation for all the queries, and $\mathbf{v}_{-i}$ as the valuations of the other advertisers. We assume for any query $q$, there is at least one advertiser that positively valuates it. Define $Q_i=\{q\in\mathcal{Q}:v_i^q>0\}$ as the query set that advertiser $i$ has positive values on. For ease of reference, in the rest of the paper, we will call the queries (keywords) that an advertiser positively valuates \emph{positive queries (keywords)}.

Denote $\mathbf{b}=(b_1,b_2,...,b_n)$ as the advertisers' bid profile, where $b_i \in R_{+}^{|\mathcal{S}|}$ is a vector indicating the $i$-th advertiser's bid prices on all the keywords in $\mathcal{S}$, and denote $\mathbf{b}_{-i}$ as the bids of advertisers excluding $i$. According to the industry practice, we assume that each advertiser can only bid on up to $\kappa$ keywords. As a result, for each $b_i$, there are at most $\kappa$ positive values. Denote $b_i^s$ as the bid price of advertiser $i$ on keyword $s$ and $\mathbf{b}^s$ as all the advertisers' bids on keyword $s$.

Based on the notations above, \textsf{SBM-GSP} can be described as follows. When a query $q$ is issued, the \textsf{SBM-GSP} mechanism first finds all the keywords that can be matched to the query. Second, it includes all the ads that bid on these keywords into the auction and uses the following formula to transform the bid prices on keywords of advertiser $i$ to his/her bid price on the query: $b_i^q=\max_{s\in N_{G}(q)} b_i^s$. In the end, the GSP auction is run upon the ads with their query-level bids, i.e., all the ads are ranked by their bids, and the payment of a clicked ad equals the bid of the ad ranked right below it.

\subsection{Solution Concepts}

In this paper, we consider rational behaviors under various assumptions on the information availablity to the advertisers. In general, the advertisers are engaged as players in a game defined by the auction mechanism (in the remaining of the paper, we use ``advertiser'' and ``player'' interchangeably). Every advertiser aims at selecting a bidding strategy that maximizes his/her utility. According to the availability of the information, we can categorize the settings into the Bayesian setting (partial information setting) and the full-information setting respectively.

In the Bayesian setting, we assume that the valuation (type) profile $\mathbf{v}$ is drawn from a publicly known distribution $\mathbf{F}$. A strategy for player $i$ is a (possibly randomized) mapping $b_i : R^{|\mathcal{Q}|}_+ \longrightarrow R^{|\mathcal{S}|}_+$, mapping his/her type $v_i$ to a bid vector $b_i(v_i)$. We use $\mathbf{b}(\mathbf{v}) = (b_1(v_1),b_2(v_2),...,b_n(v_n))$ to denote the corresponding bid profile when $\mathbf{b}(\cdot)$ is applied to $\mathbf{v}$. Denote $u_i(\mathbf{b})$ as the utility function of advertiser $i$. We say a strategy $\mathbf{b}(\cdot)$ is a Bayes-Nash equilibrium for distribution $\mathbf{F}$, if for all $i$, all $v_i$, and all alternative strategies $b'_i(\cdot)$, {\small\begin{eqnarray*}
\mathbb{E}_{\mathbf{v}_{-i},\mathbf{b}}[u_i(b_i(v_i),\mathbf{b}_{-i}(\mathbf{v}_{-i}))|v_i] \geq \mathbb{E}_{\mathbf{v}_{-i},\mathbf{b}}[u_i(b'_i(v_i),\mathbf{b}_{-i}(\mathbf{v}_{-i}))|v_i].
\end{eqnarray*}}
In other words, in a Bayes-Nash equilibrium, each player maximizes his/her expected utility using strategy $b_i(\cdot)$, assuming that the others bid according to strategies $\mathbf{b}_{-i}(\cdot)$.

In the full-information setting, the valuation profile $\mathbf{v}$ is known and fixed. In this setting, a pure strategy of any advertiser corresponds to a bid vector $b_i$. we say that a bid profile $\mathbf{b}$ is a (pure) Nash equilibrium if there is no deviation from which the players can be better off, i.e., for all advertiser $i$, for all $b'_i$,
{\small\begin{eqnarray*}
u_i(b_i,\mathbf{b}_{-i}) \geq u_i(b'_i,\mathbf{b}_{-i}).
\end{eqnarray*}}

\section{Probabilistic Broad-Match Mechanism}

As discussed in the introduction, the \textsf{SBM-GSP} mechanism has several drawbacks from a theoretical perspective. In this paper, we develop a new broad-match mechanism with better theoretical guarantee, which we call \textit{Probabilistic Broad-Match} (\textsf{PBM-GSP}) mechanism. The detail of the \textsf{PBM-GSP} mechanism is described in Algorithm \ref{alg:one}, and can be explained as below.
\begin{algorithm}[]
\SetAlgoNoLine
\KwIn{Advertiser's bid profile $\mathbf{b}$, matching probability $\pi_q(s)$ for any query $q$, keyword $s$.}
\KwOut{The ads to show for each query and the prices to charge from advertisers.}
\For{each query $q$ submitted to search engine}{
      Sample keyword $s$ according to distribution $\pi_{q}(s)$\;
      Set $b_i^q$ to be $b_i^{s}$\;
      Run the GSP auction on $b_1^q,\cdots,b_n^q$\;
      }
\caption{Probabilistic Broad-Match Mechanism with GSP Auction (\textsf{PBM-GSP})}
\label{alg:one}
\end{algorithm}

Given the query-keyword bipartite graph $G$, for each query $q\in\mathcal{Q}$, we impose a matching probability distribution $\pi_q(s)$ whose support is $N_G(q)$, i.e., $\pi_q(s)>0$ if and only if $s\in N_G(q)$, and $\sum_{s \in N_{G}(q)} \pi_q(s)=1$. With this matching probability distribution, for any issued query $q$, the mechanism randomly samples a keyword $s\in N_G(q)$, and selects the ads bidding on the keyword $s$ into the auction. For each selected ad, the bid price on keyword $s$ will be directly used as the bid price on query $q$ during this round of auction, \footnote{One may have noticed that due to the probabilistic sampling, an advertiser can only get access to a fraction of the whole query volume if he remains bidding on the same set of keywords as he/she does with \textsf{SBM-GSP}. Therefore, some advertisers may have to bid on more keywords so as to maintain the same visibility of their ads to the users. Fortunately, since the number of keywords is always significantly smaller than the number of queries, the situation will not be as serious as in exact-match mechanism.}  i.e., $b_i^q=b_i^{s}$, where $s\sim \pi_q$, and then a GSP auction is run to determine the ad allocations and prices.

For ease of description, we define $\sigma_{s,\mathbf{b}}(k)$ as the advertiser who is ranked at position $k$ and $\sigma_{s,\mathbf{b}}^{-1}(i)$ as the ranking position of advertiser $i$, for any keyword $s$ and bid profile $\mathbf{b}$. For sake of rigorousness, we define $\sigma_{s,\mathbf{b}}(k)=\infty$ if there are fewer than $k$ positive bids on keyword $s$, and define $b_{\infty}^s = v_{\infty}^q = 0$, for any query $q$ and keyword $s$. We also define $\sigma_{s,\mathbf{b}}^{-1}(i) = \infty$ if advertiser $i$ does not bid on keyword $s$, and define $w_{\infty} = 0$. Define $p_{ s,\mathbf{b}}(i)$ as the price charged to player $i$ when keyword $s$ is sampled and a user clicks on the ads, i.e., for \textsf{PBM-GSP}, if advertiser $j$ is ranked right below advertiser $i$, then $p_{ s,\mathbf{b}}(i)=b_j^s$. With the aforementioned notations, the expected utility of advertiser $i$ can be defined as
{\small\begin{eqnarray*}
u_i(\mathbf{b}) = \int_{q\in\mathcal{Q}} \sum_{s \in N_G(q)} \pi_q(s) w_{\sigma_{s,\mathbf{b}}^{-1}(i)} (v_i^q - p_{ s,\mathbf{b}}(i)) dP \end{eqnarray*}}

As a common way to rule out unnatural equilibria \cite{2012POA_GSP,2012Rev_GSP,caragiannis2011efficiency,dhangwatnotai2011multi}, we only consider conservative bidders in the theoretical analysis. It is easy to show that for any advertiser $i$ on any keyword $s$, a bidding price $ b_i^s > v_i^s $ is always weakly dominated by the bid $b_i^s = v_i^s$ (see Lemma \ref{convervativebidderlemma}), in which $v_i^s$ is the expected value of keyword $s$ for advertiser $i$ and defined as $v_i^s \triangleq\mathbb{E}[v_i^q|s]=\frac{\int_{q\in N_{G}(s)} \pi_q(s)v_i^qdP}{\int_{q\in N_{G}(s)} \pi_q(s)dP}$.
\begin{lemma}(Conservative bidder)\label{convervativebidderlemma}
For any advertiser $i$, a bid price $b^s_i> v_i^s $ for keyword $s$ is always weakly dominated by $b_i^s = v_i^s$, where $v_i^s \triangleq\mathbb{E}[v_i^q|s]=\frac{\int_{q\in N_{G}(s)} \pi_q(s)v_i^qdP}{\int_{q\in N_{G}(s)} \pi_q(s)dP}$.
\end{lemma}
\begin{proof}
Note that with the \textsf{PBM-GSP} mechanism, advertisers will not compete across keywords. For advertiser $i$, denote $u^s_i(\mathbf{b}) = \int_{q\in N_{G}(s)}\pi_q(s)w_{\sigma_{s,\mathbf{b}}^{-1}(i)}(v^q_i- p_{ s,\mathbf{b}}(i))dP$ as his/her utility obtained from keyword $s$.  It is easy to see that $u_i(\mathbf{b}) = \sum_{s\in\mathcal{S}}u^s_i(\mathbf{b})$. For any bidding profile $\mathbf{b}_{-i}$, if advertiser $i$ bids a value larger than $v_i^s$ on keyword $s$ and get the same position $k$ as bidding $v_i^s$, changing his/her bid to $v_i^s$ will not hurt his/her total utility. If he/she bids a larger value and obtains a better position $k'$, he/she will suffer a payment larger than $v^s_i$ when his/her ad is clicked, and therefore his/her expected utility on keyword $s$ must be less than
$\int_{q\in N_{G}(s)}\pi_q(s)w_{k'}[v^q_i-v^s_i]dP= w_{k'}[\int_{q\in N_{G}(s)}\pi_q(s) v^q_idP-v^s_i\int_{q\in N_{G}(s)}\pi_q(s)dP]=0$, and the theorem follows.
\end{proof}
In PBM mechanism, bids for different keywords will not be mixed up in the same auction, it is easier for advertisers to evaluate their payoffs on each keyword. As a result, they could develop more accurate bidding strategies to reflect their valuations on each keyword. For example, it can be easily shown that in single-slot setting, the dominant strategy for an advertiser is to truthfully report the expected valuation on the keyword that he/she chooses to bid.
\begin{corollary}
When there is only one slot to display Ads, for any advertiser $i$, the weakly dominant strategy for keyword $s$ is $b_i^s = v_i^s$.
\end{corollary}
In the next sections, we show this probabilistic matching can eventually improve the performances of the auction system.

\section{Social Welfare Analysis}
In this section, we present our theoretical results on the social welfare (efficiency) of the proposed \textsf{PBM-GSP} mechanism. Specifically, we study the ratio between the optimal social welfare and the worst-case welfare in equilibrium, which is also known as the Price of Anarchy (PoA) \cite{koutsoupias1999worst,giotis2008equilibria,christodoulou2008bayesian,bhawalkar2011welfare}:
\begin{itemize}
    \item \textsf{Bayes-Nash PoA} : In the Bayesian setting, we assume every advertiser $i$ privately knows his/her own valuation vector $v_i$ for the queries, and only knows a prior distribution of other advertisers' valuation vectors. Assume the valuation profile $\mathbf{v} $ is drawn from a public distribution $\mathbf{F}$ and the Bayes-Nash PoA is defined as
        {\small\begin{eqnarray*}
        \textrm{Bayes-Nash }PoA=\max_{\mathbf{F}, \mathbf{b}(\cdot)\textrm{: a Bayes-Nash equilibrium}}\frac{\mathbb{E}_{\mathbf{v}} [SW(\mathcal{OPT}(\mathbf{v}))]}{\mathbb{E}_{\mathbf{v},\mathbf{b}(\mathbf{v})}[SW(\mathbf{b}(\mathbf{v}))]},
        \end{eqnarray*}}
where $SW(\mathcal{OPT}(\mathbf{v}))$ refers to the social welfare of the optimal allocation that allocates slot $k$ of any query $q$ to the player with the $k$-th largest value, i.e.,
        {\small\begin{eqnarray}
SW(\mathcal{OPT}(\mathbf{v}))= \int_{q\in \mathcal{Q}}\sum_{k=1}^n w_k v_{[k]}^q\ dP,
        \end{eqnarray}}
where $v_{[k]}^q$ is the $k$-th largest value among the valuations of query $q$. Similarly, $SW(\mathbf{b})$ refers to the social welfare of the \textsf{PBM-GSP} mechanism with bidding profile $\mathbf{b}$, i.e.,
        {\small\begin{eqnarray}
        SW(\mathbf{b})= \int_{q\in \mathcal{Q}} \sum_{s \in N_G(q)} \pi_q(s)\sum_{k=1}^n w_k v^q_{\sigma_{s,\mathbf{b}}(k)} dP.
        \end{eqnarray}}
  \item \textsf{Pure PoA} : In the full-information setting, the valuation of each advertiser on each query is fixed and the pure PoA can be mathematically defined as follows:
        {\small\begin{eqnarray*}
        \textrm{pure }PoA=\max_{\mathbf{v}, \mathbf{b}\textrm{: a pure Nash equilibrium}}\frac{SW(\mathcal{OPT}(\mathbf{v}))}{SW(\mathbf{b})}.
        \end{eqnarray*}}
\end{itemize}
In order to characterize the influence of the maximum number of bid keywords, i.e., $\kappa$, we use \textit{expressiveness} to measure the capacity of the bidding language. The concept of expressiveness has been widely used in the literature of auction theory \cite{sandholm2007expressive,cramton2006combinatorial,lahaie2008expressive,boutilier2008expressive}, and its theoretical foundation has been established in \cite{benisch2008theory}. In this paper, we use a new notion of expressiveness, which we call the keyword-level (KL) expressiveness. As will be seen in later sections, the KL-expressiveness will affect both the social welfare and search engine revenue for the \textsf{PBM-GSP} mechanism. The formal definition of KL-expressiveness is given as below.
\begin{definition}
(Keyword-Level Expressiveness)
Given a valuation profile $\mathbf{v}$, we call the auction system $\beta$-KL-expressive, if for any advertiser $i$, $\kappa$ keywords can cover at least $\beta$ fraction of his/her positive keywords, i.e.,
$\kappa\geq \beta |\{ s: N_G(s)\cap Q_i\neq\emptyset\}|$.
We call an auction system $\beta$-KL-expressive (in the Bayesian setting), if for any valuation profile sampled from $\mathbf{F}$, the auction system is $\beta$-KL-expressive. When $\beta = 1$, we say the auction system is fully KL-expressive\footnote{In real sponsored search systems, the number of keywords that an advertiser can bid on is usually large enough to satisfy most of his/her needs. For example, in Google Adwords, advertisers are allowed to bid up to 3 million keywords, which can be regarded as quite a large number. In this case, we can consider the system as fully KL-expressive.
}.

\end{definition}

\subsection{Bayes-Nash Price of Anarchy}
In this subsection, we analyze the Bayes-Nash PoA for the \textsf{PBM-GSP} mechanism. We first extend the concept of homogeneity proposed in \cite{dhangwatnotai2011multi} to the Bayesian setting. We call the extended concept \emph{expected homogeneity}, which measures the diversity of advertisers' valuations on the queries matched to the same keyword in an expectation sense. For completeness, we list the definitions for both \emph{homogeneity} and \emph{expected homogeneity} as follows (in the full-information setting, expected homogeneity will trivially reduce to homogeneity).
\begin{definition}
(Homogeneity) \cite{dhangwatnotai2011multi}
A keyword $s$ is $c$-homogeneous if for every advertiser $i$ and two arbitrary queries $q_1,q_2\in N_G(s)$, $v_i(q_1)\leq cv_i(q_2)$. The auction system is $c$-homogeneous if every keyword $s\in\mathcal{S}$ is $c$-homogeneous.
\end{definition}
\begin{definition}
(Expected Homogeneity)
A keyword $s$ is $c$-expected-homogeneous if for any advertiser $i$, two arbitrary queries $q_1,q_2\in N_G(s)$, $P(v_i^{q_1}\leq c\mathbb{E}[v_i^{q_2}])=1$. The auction system is $c$-expected-homogeneous if every keyword $s\in\mathcal{S}$ is $c$-expected-homogeneous.
\end{definition}

We leverage the technique developed in \cite{2012POA_GSP}, which is used to analyze the PoA bound for the GSP auction.
\begin{lemma}
\cite{2012POA_GSP} We say that a game is $(\lambda,\mu)$-semi-smooth if for each player $i$ there exists some (possibly randomized) strategy $b'_i(\cdot)$ (depending only on the type of the player) such that $\sum_i \mathbb{E}_{b'_i(v_i)}[u_i(b'_i(v_i),b_{-i}] \geq \lambda \cdot SW(\mathcal{OPT}(v)) - \mu \cdot SW(b)$ holds for every pure strategy profile $b$ and every (fixed) type vector $v$ (The expectation is taken over the random bits of $b'_i(v_i)$).
If a game is $(\lambda, \mu)$-semi-smooth and its social welfare is at least the sum of the players' utilities, then the price of anarchy with uncertainty is at most $(\mu + 1)/\lambda$.
\end{lemma}

With the above definitions and lemmas, we give an upper bound for the Bayes-Nash PoA of the \textsf{PBM-GSP} mechanism.
\begin{theorem}\label{welfare1}
If the auction system is $\beta$-KL-expressive and $c$-expected-homogeneous, and the GSP auction is a $(\lambda,\mu)$-semi-smooth game, the Bayes-Nash PoA for the \textsf{PBM-GSP} mechanism is at most $c(\frac{\beta \mu + 1}{\beta\lambda })$.
\end{theorem}
To prove the theorem, we use the welfare generated from a truthfull bidding profile $\mathfrak{v}$ to connect the optimal welfare and the welfare in any Bayes-Nash equilibrium. Here the truthfull bidding profile $\mathfrak{v}$ denotes the situation when all advertisers bid their expected values on any keyword and there is no $\kappa$ constaint, i.e., $\mathfrak{v}_i^s = v_i^s$. In this situation, $SW(\mathfrak{v})$ equals $\int_{q\in \mathcal{Q}} \sum_{s \in N_G(q)} \pi_q(s)\sum_{k=1}^n w_k v^s_{[k]} dP$, where $v_{[k]}^s$ is the $k$-th largest value among all the expected valuations on keyword $s$.
\begin{proof}
We prove the theorem in two steps.
First, we bound the ratio between $\mathbb{E}_{\mathbf{v}, \mathbf{b}(\mathbf{v})}[SW(\mathbf{b}(\mathbf{v}))]$ and $\mathbb{E}_{\mathbf{v}}[SW (\mathfrak{v})]$,
and then bound the ratio between $\mathbb{E}_{\mathbf{v}}[SW (\mathfrak{v})]$ and $\mathbb{E}_{\mathbf{v}}[SW(\mathcal{OPT}(\mathbf{v}))]$. The proof details of the two steps are given below.

For the first step, we show if the GSP auction is a $(\lambda,\mu)$-semi-smooth game, for any Bayes-Nash equilibrium $\mathbf{b}(\cdot)$ of the \textsf{PBM-GSP} mechanism, the following bound holds,
{\small\begin{equation}\label{welfare-0}
\frac{\mathbb{E}_{\mathbf{v}}[SW (\mathfrak{v})]}{\mathbb{E}_{\mathbf{v}, \mathbf{b}(\mathbf{v})}[SW(\mathbf{b}(\mathbf{v}))]}<\frac{\beta \mu + 1}{\beta\lambda }.
\end{equation}}

Note that with the \textsf{PBM-GSP} mechanism, advertisers will not compete across keywords. For each advertiser $i$, define the utility on any positive keyword $s$ as $u^s_i(\mathbf{b}^s) = \int_{q\in N_{G}(s)}\pi_q(s)w_{\sigma_{s,\mathbf{b}}^{-1}(i)}(v_i^q- p_{ s,\mathbf{b}}(i) )dP$. By the defininition of $v_i^s$, this utility function can be rewritten as $u^s_i(\mathbf{b})= \int_{q\in N_{G}(s)}\pi_q(s) w_{\sigma_{s,\mathbf{b}}^{-1}(i)} (v^s_i- p_{s,\mathbf{b}}(i))dP$. Thus for this particular keyword, the advertiser's utility is exactly that for the GSP auction with true value defined as $v_i^s \int_{q\in N_{G}(s)}\pi_q(s)dP$.
Denote $S_i$ as the set of positive keywords for advertiser $i$.  Considering that the game within a given keyword is $(\lambda,\mu)$-semi-smooth, there must exist a (randomized) strategy $h_i^s(\cdot)$ on keyword $s$ satisfying, for every pure strategy $\mathbf{b}$,
{\small\begin{eqnarray}
&&\sum_{i:s\in S_i}\mathbb{E}_{h_i^s(v_i^s)}[u^s_i(h_i^s(v_i^s),\mathbf{b}^s_{-i}))] \geq \lambda SW (\mathfrak{v},s) - \mu SW(\mathbf{b},s),
\end{eqnarray}}
where $SW(\mathbf{b},s)$ is the welfare generated from keyword $s$, i.e., $SW(\mathbf{b},s) = \int_{q\in N_G(s)} \pi_q(s) \sum^n_{k=1} w_k v^s_{\sigma_{s,\mathbf{b}}(k)} dP$, and  $SW(\mathbf{b}) = \sum_{s\in \mathcal{S}} SW (\mathbf{b},s)$.

On this basis, we design a randomized strategy $b'_i(\cdot)$ for advertiser $i$ as follows. The randomized strategy $b'_i(\cdot)$ first randomly samples $\kappa$ keywords from $S_i$, and plays the strategy $h_i^s(\cdot)$ if keyword $s$ is sampled. Considering that the auction system is $\beta$-KL-expressive, the probability of any keyword $s \in S_i$ sampled by the strategy $b'_i(\cdot)$ is larger than $\beta$.

Then it is straightforward to attain
{\small\begin{eqnarray}
\mathbb{E}_{\mathbf{b}'(\mathbf{v})}[\sum_{i=1}^nu_i(b'_i (v_i),\mathbf{b} _{-i}(\mathbf{v}_{-i})))] &\geq& \beta \sum_{i=1}^n\sum_{s\in S_i}\mathbb{E}_{h^s_i(v^s_i)}[ u^s_i(h^s_i (v^s_i),\mathbf{b}^s_{-i}(\mathbf{v}_{-i})))]\nonumber\\
&=&\beta \sum_{s\in S} \sum_{i:s\in S_i} \mathbb{E}_{h^s_i(v^s_i)}[ u^s_i(h^s_i (v^s_i),\mathbf{b}^s_{-i}(\mathbf{v}_{-i}))))]\nonumber\\
&\geq& \beta \sum_{s\in S} (\lambda SW (\mathfrak{v},s) - \mu SW(\mathbf{b},s))\nonumber\\
&=&\beta\lambda SW (\mathfrak{v}) - \beta \mu SW(\mathbf{b}).
\end{eqnarray}}
Given the fact that the social welfare is at least the total utility of all the players, for any Bayes-Nash equilibrium $b(\cdot)$,we have
{\small\begin{eqnarray}
\mathbb{E}_{\mathbf{v},\mathbf{b}(\mathbf{v})}[SW(\mathbf{b}(\mathbf{v}))] &\geq&\mathbb{E}_{\mathbf{v},\mathbf{b}(\mathbf{v})}[\sum_{i=1}^n u_i(\mathbf{b}(\mathbf{v}))] \nonumber
\geq \mathbb{E}_{\mathbf{v},\mathbf{b}(\mathbf{v}),\mathbf{b}'(\mathbf{v})}[\sum_{i=1}^n \mathbb{E}_{b_i'(v_i)}[ u_i(b_i'(v_i),\mathbf{b}_{-i}(\mathbf{v}))]]\nonumber \\
&\geq& \beta\lambda\mathbb{E}_{\mathbf{v}}[ SW (\mathfrak{v})] - \beta \mu \mathbb{E}_{\mathbf{v},\mathbf{b}(\mathbf{v})}[ SW(\mathbf{b}(\mathbf{v}))]. \nonumber
\end{eqnarray}}
Then inequality (\ref{welfare-0}) follows.

For the second step, we show that $\frac{\mathbb{E}_{\mathbf{v}}[SW(\mathcal{OPT}(\mathbf{v}))]}{\mathbb{E}_{\mathbf{v}}[SW (\mathfrak{v})]} \leq c$.
Considering
{\small\begin{eqnarray}\label{welfare-1213}
\mathbb{E}_{\mathbf{v}}[SW(\mathcal{OPT}(\textbf{v}))] - c \mathbb{E}_{\mathbf{v}}[ SW (\mathfrak{v})]= \int_{q\in\mathcal{Q}}\sum_{s \in N_G(q)} \pi_q(s) (\sum_{k=1}^n \mathbb{E}_{\mathbf{v}}[w_k v_{[k]}^q] - c\sum_{k=1}^n \mathbb{E}_{\mathbf{v}}[w_k v_{[k]}^s]) dP,
\end{eqnarray}}

it suffices to prove for any keyword $s$ and any query $q \in N_G(s)$, $\sum_{k=1}^n \mathbb{E}_{\mathbf{v}}[w_k v_{[k]}^q] - c\sum_{k=1}^n \mathbb{E}_{\mathbf{v}}[w_k v_{[k]}^s] \leq 0$. Since the auction system is $c$-expected-homogeneous, the following result holds with probability one,
{\small\begin{eqnarray}
c\mathbb{E}_{\mathbf{v}}[v_i^s] &=& \mathbb{E}_{\mathbf{v}}[\frac{\int_{q'\in N_{G}(s)} c v_i^{q'}\pi_{q'}(s)dP}{\int_{q'\in N_{G}(s)} \pi_{q'}(s)dP}] = \frac{\int_{q'\in N_{G}(s)} c \mathbb{E}_{\mathbf{v}}[v_i^{q'}] \pi_{q'}(s)dP}{\int_{q'\in N_{G}(s)} \pi_{q'}(s)dP} \nonumber\\
&\geq& \frac{\int_{q'\in N_{G}(s)} v_i^q \pi_{q'}(s)dP}{\int_{q'\in N_{G}(s)} \pi_{q'}(s)dP} = v_i^q.
\end{eqnarray}}

Without loss of generality, we assume that $\mathbb{E}_{\mathbf{v}}[v_1^s] \geq \mathbb{E}_{\mathbf{v}}[v_2^s] \geq \cdots \mathbb{E}_{\mathbf{v}}[v_n^s]$ for keyword $s$. Then we have
{\small\begin{eqnarray}\label{welfare-1212}
\sum_{k=1}^n \mathbb{E}_{\mathbf{v}}[w_k v_{[k]}^q] - c\sum_{k=1}^n \mathbb{E}_{\mathbf{v}}[w_k v_{[k]}^s]&\leq& c\sum_{k=1}^n \mathbb{E}_{\mathbf{v}}[w_k \mathbb{E}_{\mathbf{v}}[v_k^s]] - c\sum_{k=1}^n \mathbb{E}_{\mathbf{v}}[w_k v_{[k]}^s]\nonumber \\
&=& c\mathbb{E}_{\mathbf{v}}[\sum_{k=1}^nw_k v_k^s - \sum_{k=1}^n w_k v_{[k]}^s]\leq 0.
\end{eqnarray}}
Applying (\ref{welfare-1212}) to (\ref{welfare-1213}), we can prove $\frac{\mathbb{E}_{\mathbf{v}}[SW(\mathcal{OPT}(\mathbf{v}))]}{\mathbb{E}_{\mathbf{v}}[SW (\mathfrak{v})]} \leq c$. Then the theorem follows by combining the two steps.
\end{proof}
In \cite{2012POA_GSP}, it is shown that the GSP auction is $(1-\frac{1}{e},1)$-semi-smooth. Furthermore, it is trivial to obtain that the GSP auction in the single-slot case is a $(1,1)$-semi-smooth game. Therefore, we can obtain the following two corollaries.
\begin{corollary}
If the auction system is $\beta$-KL-expressive and $c$-expected-homogeneous, the Bayes-Nash PoA for the \textsf{PBM-GSP} mechanism is at most $\frac{e}{e-1}\frac{\beta+1}{\beta}c$.
\end{corollary}
\begin{corollary}\label{welfare10}
If the auction system is $\beta$-KL-expressive and $c$-expected-homogeneous and there is only one slot to display ads, the Bayes-Nash PoA for the \textsf{PBM-GSP} mechanism is at most $\frac{\beta+1}{\beta}c$.
\end{corollary}

\subsection{Pure Price of Anarchy in Full-Information Setting}

In this subsection, we analyze the pure PoA for the \textsf{PBM-GSP} mechanism.
In particular, based on the notions of KL-expressiveness and homogeneity, we derive the following pure PoA bound.
\begin{theorem}\label{thm_POA_FULL}
If the auction system is $\beta$-KL-expressive and $c$-homogeneous, the pure PoA of \textsf{PBM-GSP} mechanism for the multi-slot case is at most $\frac{\beta+1}{\beta}c$.
\end{theorem}
\begin{proof}
Similar to Theorem \ref{welfare1}, the proof of Theorem \ref{thm_POA_FULL} contains two steps. For the first step, we prove $\max_{\mathbf{v}, \mathbf{b}:\textrm{a pure Nash equilibrium}}\frac{SW (\mathfrak{v})}{SW(\mathbf{b})}<\frac{1}{\beta}+1$.

Denote $O_i$ as the set of (keyword, position) pair that advertiser $i$ wins when all advertisers truthfully bid, i.e., $O_i=\{(s,k) :\sigma_{s,\mathfrak{v}}(k) = i, w_k>0\}$, denote $S_i$ as the keyword set in $O_i$. Given any bid profile $\mathbf{b}$, denote $O'_i$ as the (keyword, position) set that advertiser $i$ actually bids and wins, i.e., $O'_i=\{(s,k):\sigma_{s,\mathbf{b}}(k)=i,w_k>0\}$, denote $S'_i$ as the set of keywords in $O'_i$, whose size is no larger than $\kappa$.

We divide advertisers into three categories, $I_1, I_2, I_3$: (1) advertisers in $I_1$ bid on $\kappa$ keywords and $S_i\setminus S_i'\neq\emptyset$; (2) advertisers in $I_2$ bid on $\kappa$ keywords and $S_i\setminus S_i'=\emptyset$; (3) advertisers in $I_3$ bid on fewer than $\kappa$ keywords. We apply the equilibrium conditions to the three categories respectively.

1.~For any advertiser $i$ in category $I_1$, by definition, advertiser $i$ wins a position in any keyword in $S'_i$.

So, first, advertiser $i$ will not increase his/her payoff by changing his/her strategy from bidding a keyword with position $(s',k')\in O'_i$ to any keyword $s\in S_i\setminus S'_i$ with position $k$, where $(s,k)\in O_i$. Considering all advertisers are conservative, we have
{\small\begin{eqnarray}\label{pure-1}
\int_{q\in N_G(s')}\pi_q(s')w_{k'} v_i^q dP &\geq& \int_{q\in N_G(s')}\pi_q(s')w_{k'}(v_i^q-p_{s,\mathbf{b}}(i)) dP \nonumber\\
&\geq& \int_{q\in N_G(s)}\pi_q(s)w_kv_i^q dP - \int_{q\in N_G(s)}\pi_q(s)w_k b_{\sigma_{s,\mathbf{b}}(k)}^s dP \nonumber\\
 &\geq& \int_{q\in N_G(s)}\pi_q(s)w_kv_i^q dP - \int_{q\in N_G(s)}\pi_q(s)w_k v_{\sigma_{s,\mathbf{b}}(k)}^s dP.
\end{eqnarray}}
Summing up both sides over all advertisers $i\in I_1$, $(s,k)\in O_i$ where $s\in S_i\setminus S'_i$ and $(s',k')\in O'_i$, we have
{\small\begin{eqnarray}\label{full-inf-proof-1}
&&\kappa\sum_{i\in I_1}\sum_{(s,k):s\in S_i\setminus S'_i,(s,k)\in O_i}\int_{q\in N_G(s)}\pi_q(s)w_kv_i^q dP \nonumber\\
\leq&& \sum_{i\in I_1}\sum_{(s,k):s\in S_i\setminus S'_i,(s,k)\in O_i}\sum_{(s',k'):(s',k')\in O'_i}\int_{q\in N_G(s')}\pi_q(s')w_{k'}v_i^qdP\nonumber\\
&+&\kappa\sum_{i\in I_1}\sum_{(s,k):s\in S_i\setminus S'_i,(s,k)\in O_i}\int_{q\in N_G(s)}\pi_q(s)w_k v_{\sigma_{s,\mathbf{b}}(k)}^s dP.
\end{eqnarray}}
Second, advertiser $i$ will not increase his/her payoff by changing his/her strategy from bidding on keyword $s$ with position $k'$ to bidding the same keyword with position $k$, where $(s,k')\in O'_i$, $(s,k)\in O_i$, and $s\in S_i\cap S'_i$. Similar to (\ref{pure-1}), we have
{\small\begin{eqnarray}
\int_{q\in N_G(s)}\pi_q(s)w_kv_i^q dP&\leq& \int_{q\in N_G(s)}\pi_q(s)w_{k'}v_i^qdP  + \int_{q\in N_G(s)}\pi_q(s)w_k v_{\sigma_{s,\mathbf{b}}(k)}^s dP.
\end{eqnarray}}
Summing up both sides over all advertisers $i\in I_1$, $(s,k)\in O_i$ and $(s,k')\in T'_i$ where $s\in S_i\cap S'_i$, we have
{\small\begin{eqnarray}\label{full-inf-proof-2}
\sum_{i\in I_1}\sum_{(s,k):s\in S_i\cap S'_i,(s,k)\in O_i}\int_{q\in N_G(s) }\pi_q(s)w_kv_i^q dP\leq&&\sum_{i\in I_1}\sum_{(s',k'):s'\in S_i\cap S'_i,(s',k')\in O'_i}\int_{q\in N_G(s')}\pi_q(s')w_{k'}v_i^qdP \nonumber\\
&+&\sum_{i\in I_1}\sum_{(s,k):s\in S_i\cap S'_i,(s,k)\in O_i}\int_{q\in N_G(s)}\pi_q(s)w_k v_{\sigma_{s,\mathbf{b}}(k)}^s dP.
\end{eqnarray}}
Summing up (\ref{full-inf-proof-1}) and $\kappa$ times (\ref{full-inf-proof-2}), we have
{\small\begin{eqnarray}
\kappa\sum_{i\in I_1}\sum_{(s,k)\in O_i}\int_{q\in N_G(s)}\pi_q(s) w_kv_i^q dP
\leq&& \kappa\sum_{i\in I_1}\sum_{(s',k'):s'\in S_i\cap S'_i,(s',k')\in O'_i}\int_{q\in N_G(s')}\pi_q(s') w_{k'}v_i^qdP\nonumber\\
&+&\sum_{i\in I_1}\sum_{(s,k):s\in S_i\setminus S'_i,(s,k)\in O_i}\sum_{(s',k'):(s',k')\in O'_i}\int_{q\in N_G(s')}\pi_q(s') w_{k'}v_i^qdP\nonumber\\
&+&\kappa\sum_{i\in I_1}\sum_{(s,k)\in O_i}\int_{q\in N_G(s)}\pi_q(s) w_k v_{\sigma_{s,\mathbf{b}}(k)}^s dP.
\end{eqnarray}}
Considering that $|S_i\setminus S'_i|\leq \frac{\kappa}{\beta}-\kappa$, and $\{(s',k'):s'\in S_i\cap S'_i,(s',k')\in O'_i\}\subset \{(s',k'):(s',k')\in O'_i\}$, we obtain
{\small\begin{eqnarray}\label{full-inf-proof-3}
\kappa\sum_{i\in I_1}\sum_{(s,k)\in O_i}\int_{q\in N_G(s)}\pi_q(s)w_kv_i^q dP\leq&& \frac{\kappa}{\beta}\sum_{i\in I_1}\sum_{(s',k')\in O'_i}\int_{q\in N_G(s')}\pi_q(s')w_{k'}v_i^qdP\nonumber\\
&+&\kappa\sum_{i\in I_1}\sum_{(s,k)\in O_i}\int_{q\in N_G(s)}\pi_q(s)w_k v_{\sigma_{s,\mathbf{b}}(k)}^s dP.
\end{eqnarray}}
2.~For advertiser $i$ in category $I_2$ and $I_3$, since $\mathbf{b}$ is a Nash equilibrium, it is clear that $S_i\subset S'_i$, and for $\forall s\in S_i$, $(s,k)\in O_i$, $(s,k')\in O'_i$, the following holds,
{\small\begin{eqnarray}
\int_{q\in N_G(s)}\pi_q(s)w_kv_i^q dP&\leq& \int_{q\in N_G(s)}\pi_q(s)w_{k'}v_i^qdP  + \int_{q\in N_G(s)}\pi_q(s)w_k v_{\sigma_{s,\mathbf{b}}(k)}^s dP.
\end{eqnarray}}
By summing over all advertisers $i\in I_2\cup I_3$, $(s,k)\in O_i$ and $(s,k')\in O'_i$, we have
{\small\begin{eqnarray}\label{full-inf-proof-4}
\kappa\sum_{i\in I_2\cup I_3}\sum_{(s,k)\in O_i}\int_{q\in N_G(s)}\pi_q(s)w_kv_i^q dP\leq && \kappa\sum_{i\in I_2\cup I_3}\sum_{(s',k'):(s',k')\in O'_i}\int_{q\in N_G(s')}\pi_q(s')w_{k'}v_i^qdP\nonumber\\
&+&\kappa\sum_{i\in I_2\cup I_3}\sum_{(s,k)\in O_i}\int_{q\in N_G(s)}\pi_q(s)w_k v_{\sigma_{s,\mathbf{b}}(k)}^s dP.
\end{eqnarray}}
Since $ SW (\mathfrak{v}) = \sum_{i}\sum_{(s,k)\in O_i}\int_{q\in N_G(s)}\pi_q(s)w_kv_i^q dP$, by summing (\ref{full-inf-proof-3}) and (\ref{full-inf-proof-4}) together, we obtain the following inequality and thus complete the first step.
{\small\begin{eqnarray}
\kappa SW (\mathfrak{v})&\leq&\frac{\kappa}{\beta}SW(\mathbf{b})+\kappa SW(\mathbf{b}).
\end{eqnarray}}
For the second step, it is easy to show $\frac{SW(\mathcal{OPT}(\mathbf{v}))}{ SW (\mathfrak{v})}\leq c$ still holds in the full-information setting. Then by combining the two steps, we prove the theorem.
\end{proof}
Note that the pure PoA bound given by the above theorem is for the general multi-slot case. The result can be further optimized if we are only interested in the single-slot case (see the following theorem). We leave the proof of the theorem to the Appendix.
\begin{theorem}\label{welfare21}
If the auction system is $\beta$-KL-expressive and $c$-homogeneous, and there is only one slot to display ad, the pure PoA for the \textsf{PBM-GSP} mechanism is at most $\frac{c}{\beta}$, and the bound is tight with respect to the factors.
\end{theorem}

\begin{proof}
It is easy to see for any advertiser $i$, bidding the expected value $v_i^s$ on $s$ is the dominant strategy if he/she bids $s$ in single-slot setting, thus it suffices to consider the equilibria in which each winner of each keyword bid the true value. Similar to Theorem \ref{thm_POA_FULL}, our proof contains two steps.

For the first step, we prove that for an arbitrarily given valuation profile $\mathbf{v}$,
{\small\begin{equation}
\max_{\mathbf{b}:\textrm{a pure Nash equilibrium}}\frac{SW (\mathfrak{v})}{SW(\mathbf{b})}<\frac{1}{\beta}.
\end{equation}}

Denote $S_i$ as the keyword set that advertiser $i$ wins when all advertisers truthfully bid, i.e., $S_i=\{s :\sigma_{s,\mathfrak{v}}(1)=i\}$. Given any bid profile $\mathbf{b}$, denote $S'_i$ as the keyword set that advertiser $i$ actually bids on, i.e., $S'_i=\{s:b_i^s>0\}$, whose size is no larger than $\kappa$.

We divide advertisers into three categories, $I_1,I_2,I_3$: (1) advertisers in $I_1$ bid on $\kappa$ keywords and $S_i\setminus S_i'\neq\emptyset$; (2) advertisers in $I_2$ bid on $\kappa$ keywords and $S_i\setminus S_i'=\emptyset$; (3) advertisers in $I_3$ bid on fewer than $\kappa$ keywords. We apply the equilibrium conditions to the three categories respectively.

1.~For any advertiser $i$ in category $I_1$, since $S_i\setminus S_i'\neq\emptyset$, it is easy to show that advertiser $i$ wins all keywords in $S'_i$ (otherwise, alternatively bidding on a keyword in $S_i\setminus S'_i$ will lead to a better payoff), which yields
{\small\begin{equation}\label{appendix-full-info-welf-07}
\sigma_{s',\mathbf{b}}(1)=i, b_i^{s'} = v_i^{s'}. \forall s'\in S'_i, \forall i\in I_1.
\end{equation}}
Moreover, advertiser $i$ will not increase his/her payoff by changing his/her strategy from bidding on $s'\in S'_i$ to any $s\in S_i\setminus S'_i$, that is,
{\small\begin{eqnarray}\label{appendix-full-info-welf-08}
\int_{q\in N_G(s)}\pi_q(s)(v_i^q-b_{\sigma_{s,\mathbf{b}}(1)}^s)dP&\leq& \int_{q\in N_G(s')}\pi_q(s')(v_i^q-p_{s,\mathbf{b}}(i))dP.
\end{eqnarray}}
By dropping $ p_{s,\mathbf{b}}(i)$ from the RHS of (\ref{appendix-full-info-welf-08}) which is non-negative and using the fact that $b_{\sigma_{s,\mathbf{b}}(1)}^s \leq v_{\sigma_{s,\mathbf{b}}(1)}^s$, we have
{\small\begin{eqnarray*}
\int_{q\in N_G(s)}\pi_q(s)v_i^qdP&\leq& \int_{q\in N_G(s')}\pi_q(s')v_i^qdP + \int_{q\in N_G(s)}\pi_q(s)v_{\sigma_{s,\mathbf{b}}(1)}^q dP.
\end{eqnarray*}}
Summing up both sides over all advertisers $i\in I_1$, $s\in S_i\setminus S'_i$ and $s'\in S'_i$, we have,
{\small\begin{eqnarray}\label{appendix-full-info-welf-09}
\kappa\sum_{i\in I_1}\sum_{s\in S_i\setminus S'_i}\int_{q\in N_G(s)}\pi_q(s)v_i^qdP&\leq& \sum_{i\in I_1}\sum_{s\in S_i\setminus S'_i}\sum_{s'\in S'_i}\int_{q\in N_G(s')}\pi_q(s')v_i^qdP\nonumber\\
&+&\kappa\sum_{i\in I_1}\sum_{s\in S_i\setminus S'_i}\int_{q\in N_G(s)}\pi_q(s)v_{\sigma_{s,\mathbf{b}}(1)}^q dP.
\end{eqnarray}}
Considering that $|S_i\setminus S'_i|\leq \frac{\kappa}{\beta}-\kappa$, we apply (\ref{appendix-full-info-welf-07}) to the first term in the RHS of (\ref{appendix-full-info-welf-09}), which yields,
{\small\begin{eqnarray}\label{appendix-full-info-welf-10}
\kappa\sum_{i\in I_1}\sum_{s\in S_i\setminus S'_i}\int_{q\in N_G(s)}\pi_q(s)v_i^q dP
&\leq& (\frac{\kappa}{\beta}-\kappa)\sum_{i\in I_1}\sum_{s'\in S'_i}\int_{q\in N_G(s')}\pi_q(s')v_{\sigma_{s',\mathbf{b}}(1)}^q dP \nonumber\\
&+&\kappa\sum_{i\in I_1}\sum_{s\in S_i\setminus S'_i}\int_{q\in N_G(s)}\pi_q(s) v_{\sigma_{s,\mathbf{b}}(1)}^q dP.
\end{eqnarray}}
Since $\sum_{i\in I_1}\sum_{s'\in S'_i}\int_{q\in N_G(s')}\pi_q(s')v_{\sigma_{s',\mathbf{b}}(1)}^q dP\leq SW(\mathbf{b})$, by adding $\kappa\sum_{i\in I_1}\sum_{s\in S_i\cap S'_i}\int_{q\in N_G(s)}\pi_q(s)v_i^qdP$ to both sides of (\ref{appendix-full-info-welf-10}), we have
{\small\begin{eqnarray}\label{appendix-full-info-welf-11}
\kappa\sum_{i\in I_1}\sum_{s\in S_i}\int_{q\in N_G(s)}\pi_q(s)v_i^qdP
\leq (\frac{\kappa}{\beta}-\kappa)SW(\mathbf{b}) +\kappa\sum_{i\in I_1}\sum_{s\in S_i}\int_{q\in N_G(s)}\pi_q(s)v_{\sigma_{s,\mathbf{b}}(1)}^q dP.
\end{eqnarray}}

2.~For advertiser $i$ in category $I_2$ and $I_3$, since $\mathbf{b}$ is a Nash equilibrium, it is clear that $S_i\subset S'_i$, and for $\forall s\in S_i$, $\sigma_{s,\mathbf{b}}(1)=i$. Therefore
{\small\begin{equation}\label{appendix-full-info-welf-12}
\kappa\sum_{i\in I_2\cup I_3}\sum_{s\in S_i}\int_{q\in N_G(s)}\pi_q(s)v_i^q dP=\kappa\sum_{i\in I_2\cup I_3}\sum_{s\in S_i}\int_{q\in N_G(s)}\pi_q(s)v_{\sigma_{s,\mathbf{b}}(1)}^q dP.
\end{equation}}
According to the definitions of $SW(\mathfrak{v})$ and $SW(\mathbf{b})$, we can prove $\max_{\mathbf{b}\textrm{:a pure Nash equilibrium}}\frac{SW (\mathfrak{v})}{SW(\mathbf{b})}<\frac{1}{\beta}$ by summing up (\ref{appendix-full-info-welf-11}) (\ref{appendix-full-info-welf-12}) together.
\\

For the second step, we have $\frac{SW(\mathcal{OPT}(\mathbf{v}))}{SW(\mathfrak{v})}\leq c$. Then by combining the two steps, we prove the theorem.
\end{proof}

\subsection{Comparison between \textsf{PBM-GSP} and \textsf{SBM-GSP}}
In this subsection, we make comparisons between \textsf{PBM-GSP} and \textsf{SBM-GSP}. The overall conclusion is that the \textsf{PBM-GSP} mechanism has a better social welfare in equlibirium than the \textsf{SBM-GSP} mechanism. The detailed analysis is given as follows.

To the best of our knowledge, the theoretical analysis on $\textsf{SBM-GSP}$ \cite{dhangwatnotai2011multi} only covers the welfare in the full-information setting and the single-slot case. Therefore, we will compare \textsf{PBM-GSP} with \textsf{SBM-GSP} in this setting. Furthermore, in \cite{dhangwatnotai2011multi}, the same definition of \textit{homogeneity} but a different definition of \textit{expressiveness} is used. To avoid confusions, we refer to the expressiveness defined in \cite{dhangwatnotai2011multi} as \textit{Query-Level (QL) Expressiveness}, whose definition is copied as follows.

\begin{definition}\cite{dhangwatnotai2011multi}
(QL-Expressiveness)
We call an auction system $\alpha$-QL-expressive, if for any advertiser $i$, and any query set $Q$ satisfying $Q\subset Q_i$, $|Q|\leq \alpha |Q_i|$, there always exist $\kappa$ keywords that can cover $Q$ through the query-keyword bipartite graph $G$. When $\alpha=1$, we say the auction system is fully QL-expressive.
\end{definition}
Based on the above concepts, an almost-tight pure PoA bound for the \textsf{SBM-GSP} mechanism in the single-slot case is derived in \cite{dhangwatnotai2011multi}, as shown below.
\begin{theorem}\cite{dhangwatnotai2011multi}
If the auction system is $\alpha$-QL-expressive and $c$-homogeneous, the pure PoA of the \textsf{SBM-GSP} mechanism is at most $\frac{c^2+c}{\alpha}$.
\end{theorem}
If we compare this PoA bound with the corresponding PoA bound of the \textsf{PBM-GSP} mechanism, we will have the following discussions.

First, since different notions of expressiveness are used, if we want to compare the bounds, we need to characterize the relationship between KL-expressiveness and QL-expressiveness. Actually, a natural question is why not also using the QL-expressiveness to analyze the theoretical properties of \textsf{PBM-GSP}. The following example, which shows that the pure PoA of the \textsf{PBM-GSP} mechanism could be irrelevant to QL-expressiveness, justifies the necessity to introduce the concept of KL-expressiveness.
\begin{example}
Suppose there is only one advertiser who is allowed to bid on at most one keyword. Consider there is a fixed set of positive queries, and each query is matched to a shared keyword and other $N$ different keywords. We consider a \textsf{PBM-GSP} mechanism that matches a query to keywords with uniform probability. In this case, the auction system is always fully QL-expressive since the advertiser can use the shared keyword to reach all queries, but the welfare in equilibrium can be arbitrarily bad in \textsf{PBM-GSP} as $N$ approaches infinity.
\end{example}

Furthermore, according to our theoretical and empirical studies (details are given in the Appendix), given the query-keyword bipartite graph, the QL-expressiveness $\alpha$ and KL-expressiveness $\beta$ are actually comparable in their values (i.e., they only differ by a small constant). Therefore the difference in these two notions of expressiveness should not affect the comparison between the two PoA bounds by much.

Second, the two PoA bounds have different orders with respect to the homogeneity $c$. As aforementioned, homogeneity describes the diversity of advertiser's valuations on different queries matched to the same keyword. Take the keyword ``spider'' as an example. It can be matched to multiple queries, such as ``spider movie'', ``get rid of spider'', and ``crystal spider'', which have quite different semantic meanings. If each advertiser is only interested in one type of these semantic meanings, the homogeneity quantity $c$ will be very large due to the high valuations on some queries and the low valuations on the other queries. In this case, different orders of $c$ will lead to significant difference in the overall PoA bounds. In particular, the pure PoA bound of \textsf{PBM-GSP} is much better than that of \textsf{SBM-GSP}, since the former is linear to $c$ but the latter is quadratic.

\section{Revenue Analysis}
In this section, we study the Bayes-Nash revenue with reserve price \cite{1981myerson,dhangwatnotai2010revenue, hartline2010bayesian} for the \textsf{PBM-GSP} mechanism. We show that with a naturally-defined reserve price $r_s$ on each keyword $s$, the worst-case ratio between the optimal social welfare and the revenue of \textsf{PBM-GSP}, defined as below, can be upper bounded.
$$\max_{\mathbf {v}, \mathbf{b}(\cdot)\textrm{: a Bayes-Nash equilibrium}}\frac{\mathbb{E}_{\mathbf{v}}[SW(\mathcal{OPT}(\mathbf{v}))]}{\mathbb{E}_{\mathbf{v}, \mathbf{b}(\mathbf{v})}[\mathcal{R}_r(\mathbf{b}(\mathbf{v}))]},$$
where $\mathcal{R}_r(\mathbf{b})$ presents the revenue with bid $\mathbf{b}$ and reserve price vector $r$, i.e.,$\mathcal{R}_r(\mathbf{b})= \int_{q\in \mathcal{Q}} \sum_{s \in N_G(q)} \pi_q(s)\sum_{k=1}^n w_k \max\{r_s,b^s_{\sigma_{s,\mathbf{b}}(k+1)}\}\mathbb{I}[ b^s_{\sigma_{s,\mathbf{b}}(k)}\geq r_s] dP$.

In this paper, we assume that the auctioneer (search engine) has a public prior distribution $F$ on any advertiser's valuation vector, and any advertiser $i$'s valuation vector $v_i$ is i.i.d. sampled from this distribution, i.e., $\mathbf{F} = F^n$. Considering that for any advertiser, the valuation on keywords are the weighted averages of the valuations on the queries that the keyword can be matched to, it could be easily proved that an advertiser's expected valuation vector on keywords can also be considered as i.i.d. sampled. Thus we define, for any advertiser, the keyword valuation vector is i.i.d sampled from a distribution $T$ (induced from $F$ and the mechanism), and define the (marginal) cumulate distribution function of valuation on keyword $s$ as $T_s$ and the probability density function on keyword $s$ as $t_s$. As in common practice, we consider a particular class of distributions for $T$, which is called monotone hazard rate (MHR) distribution \cite{dhangwatnotai2010revenue,hartline2009simple,lucier2012revenue}.

For the reserve price, we employ a naturally-defined reserve price vector $r$, which is a direct extension of the Myerson reserve price  \cite{lucier2012revenue,chawla2007algorithmic,1981myerson} : For any keyword $s$, the reserve price $r_s$ is the Myerson reserve price, which satisfies $\phi_s(r_s) = 0$, where $\phi_s(v)$ is defined as the virtual value of any type $v$ on keyword $s$, i.e., $\phi_s(v) \triangleq v - \frac{1-T_s(v)}{t_s(v)}$.

Different from the GSP auction with the Myerson reserve price, it is easy to construct an example to show that, even in a single-slot case, when there is a constraint on the total number of bid keywords (i.e., $\kappa$), the ratio between the revenue of \textsf{PBM-GSP} with the Myerson reserve price and optimal social welfare can become arbitrarily bad.

\begin{example}\label{example}
The high-level idea of the example is to construct a case in which there are only two keywords, a single slot, and one advertiser, and the advertiser can only bid on one keyword. We assume that one of the keywords guarantees high welfare but low utility for the advertiser, and the other guarantees low welfare but high utility. Thus in any Bayes-Nash equilibrium, the revenue will be very low compared to the optimal welfare since the advertiser is likely to bid on the low-welfare high-utility keyword (thus low revenue to search engine) to be better off.

Denote $q_1$ and $q_2$ as two queries, and denote $s_1$ and $s_2$ as two keywords. We assume that the two queries are equally likely to be issued, and assume that in the query-keyword bipartite graph $s_1$ is only matched to $q_1$ and $s_2$ is only matched to $q_2$. Assume there is only one advertiser who is only allowed to bid on one keyword, and the advertiser's valuation on each query (keyword) is independent. Considering the one-to-one mapping between queries and keywords, we use \textit{query} and \textit{keyword} interchangeably in the following descriptions.

Let $\epsilon_1>0$, $\epsilon_2>0$, and $M>0$ be some fixed values satisfying $\epsilon_1< \frac{1}{10}$, $\epsilon_2< \frac{\epsilon_1}{10}$, and $M>10$. For keyword  $s_1$, denote the valuation density function as $t_{1}(x)$ which is supported on $ [0,2\epsilon_1 - \epsilon_1^2]$; For keyword $s_2$, the valuation density function $t_{2}(x)$ is supported on $ [0,2^M]$. We assume $t_{1}(x)$ has the following properties: (1) $t_{1}(x)$ is an increasing and differentiable function; (2) $t_{1}(x) = \frac{1}{\epsilon_1}$ when $x \in [\epsilon_1, 2\epsilon_1 - \epsilon_1^2)$. We assume $t_{2}(x)$ has the following proprieties: (1) $t_{2}(x)$ is an increasing and differentiable function; (2) $t_{2}(2^M-\epsilon_2) = \frac{1}{2^{M+1}}$; (3) $t_{2}(2^M-\frac{\epsilon_2}{2}) = \frac{1}{2^M}$; (4) $\int_{2^M-\epsilon_2}^{2^M} t_{2}(x) dx = 1 - \epsilon_2$. It is easy to check the existence of such probability density functions.

First, we show under the above conditions, the expected homogeneity value $c$ is smaller than a constant. For query $q_1$, since $\epsilon_1 <\frac{1}{10}$, we have
{\small\begin{eqnarray}
\mathbb{E}[v_{q_1}] \geq \mathbb{P} (v_{q_1}>\epsilon_1)\mathbb{E}[v_{q_1}|v_{q_1}>\epsilon_1]= \int_{\epsilon_1}^{2\epsilon_1 - \epsilon_1^2} \frac{1}{\epsilon_1} x dx
= (1 - \epsilon_1)\frac{\epsilon_1 + (2\epsilon_1 - \epsilon_1^2)}{2} \geq \epsilon_1 \geq \frac{1}{2} (2\epsilon_1 - \epsilon_1^2). \nonumber
\end{eqnarray}}
For query $s_2$, considering $\epsilon_2 < \frac{1}{10}$ and $M>10$, we have
{\small\begin{eqnarray}
\mathbb{E}[v_{s_2}] \geq \int_{2^M-\epsilon_2}^{2^M} t_{2}(x) x dx \geq (1 - \epsilon_2)(2^M-\epsilon_2) \geq \frac{1}{2} 2^M.
\end{eqnarray}}
As a consequence, we come to the conclusion that $c$ is always bounded by 2.

Second, we show that $r_{s_1}<\epsilon_1$ and $ 2^M - \epsilon_2 < r_{s_2} <2^M - \frac{\epsilon_2}{2} $. Considering the value distribution is a MHR distribution, it suffices to prove $\phi_1(\epsilon_1)>0$ and $ \phi_2(2^M - \epsilon_2)<0<\phi_2(2^M - \frac{\epsilon_2}{2})$, which can be directly obtained from below.
{\small\begin{eqnarray}
\phi_1(\epsilon_1)=\epsilon_1 - \frac{1-T_{1}(\epsilon_1)}{t_{1}(\epsilon_1)}& =&\epsilon_1 - \frac{1 - \epsilon_1}{\frac{1}{\epsilon_1}} = \epsilon_1^2 > 0\\
\phi_2(2^M - \epsilon_2) = 2^M -\epsilon_2 - \frac{1-T_{2}(2^M - \epsilon_2)}{t_{2}(2^M - \epsilon_2)} &=& 2^M - \epsilon_2 - 2^{M+1}(1- \epsilon_2)< 0 \\
\phi_2(2^M - \frac{\epsilon_2}{2}) = 2^M -\frac{\epsilon_2}{2} - \frac{1-T_{2}(2^M - \frac{\epsilon_2}{2})}{t_{2}(2^M - \frac{\epsilon_2}{2})} &>& 2^M - \frac{\epsilon_2}{2} - 2^{M}(1 - \epsilon_2)> 0
\end{eqnarray}}
Finally, we give a lower bound of welfare and an upper bound of revenue in Bayes-Nash equilibrium. Since the reserve price on keyword $s_1$ is smaller than $\epsilon_1$ while the reserve price on keyword $s_2$ is larger than $2^M - \epsilon_2$, if the valuation $v^{s_1}$ is larger than $\epsilon_1 + \epsilon_2$, the utility of advertiser will be larger than $\epsilon_2$, and bidding keyword $s_1$ will be the dominant strategy. As a consequence, we have
{\small\begin{eqnarray}
\mathbb{E}[\mathcal{R}_r(\mathbf{b}(\mathbf{v}))] \leq  \frac{1}{2}\mathbb{P}(v^{s_1} > (\epsilon_1 + \epsilon_2)) r_{s_1} + \frac{1}{2}\mathbb{P}(v^{s_1} \leq (\epsilon_1 + \epsilon_2)) r_{s_2} \leq (1 - \epsilon_1 - \frac{\epsilon_2}{\epsilon_1})\frac{\epsilon_1}{2} + (\epsilon_1 + \frac{\epsilon_2}{\epsilon_1}) 2^{M-1}.
\end{eqnarray}}
The first inequality holds since $r_{s_1}<r_{s_2}$. Similarly, we can lower bound the expected optimal welfare by
{\small\begin{eqnarray}
\mathbb{E}[SW(\mathcal{OPT}(\mathbf{v}))] \geq \frac{1}{2}\mathbb{P}(v^{s_1} > r_{s_1}) r_{s_1} + \frac{1}{2}\mathbb{P}(v^{s_2} > r_{s_2}) r_{s_2} \geq (1-\epsilon_2 - \frac{\epsilon_2}{2}\frac{1}{2^M})\frac{2^M - \epsilon_2}{2}.
\end{eqnarray}}
Fixing $M$ and letting $\epsilon_1$ and $\frac{\epsilon_2}{\epsilon_1}$ approach zero, we have
{\small\begin{eqnarray}
\lim_{\epsilon_1\rightarrow 0, \frac{\epsilon_2}{\epsilon_1} \rightarrow 0}\frac{\mathbb{E}[\mathcal{R}_r(\mathbf{b}(\mathbf{v}))]}{\mathbb{E}[SW(\mathcal{OPT})(\mathbf{v})]} \leq \lim_{\epsilon_1\rightarrow 0, \frac{\epsilon_2}{\epsilon_1} \rightarrow 0}\frac{(1 - \epsilon_1 - \frac{\epsilon_2}{\epsilon_1})\frac{\epsilon_1}{2} + (\epsilon_1 + \frac{\epsilon_2}{\epsilon_1}) 2^{M-1}}{ (1-\epsilon_2 - \frac{\epsilon_2}{2}\frac{1}{2^M})\frac{2^M - \epsilon_2}{2}}=0
\end{eqnarray}}
\end{example}

To obtain meaningful results, we consider $\phi_s(v)$  with Lipchitz condition, which is defined as below.

\begin{definition} (MHR distribution with bounded derivative)
We say a distribution $T$ is an MHR distribution with bounded derivative $\eta$, if for any keyword $s$, the following conditions hold: (1) $\phi'_s(v) \geq 1$, for all $v\geq 0$, (2) $\phi'_s(v) \leq \eta$,  for all $v \geq r_s$.
\end{definition}
The following theorem shows that when the distribution is MHR with bounded derivative, we can obtain a bound for the ratio between optimal social welfare and worst-case revenue.
\begin{theorem}\label{revenue-2}
If any advertiser's keyword valuation vector is i.i.d. drawn from an MHR distribution $T$ with bounded derivative $\eta$, the auction system is $\beta$-KL-expressive and $c$-expected-homogeneous, then the revenue obtained by the \textsf{PBM-GSP} mechanism with the Myerson reserve price is at least $\frac{\beta}{1+\beta} \frac{1}{2\eta (c e)^2}$ of the optimal social welfare.
\end{theorem}

The theorem can be proved in three steps.
First, we use Lemma \ref{revenue-0} to bound the ratio between the revenue and the welfare of \textsf{PBM-GSP} with the Myerson reserve price. Second, we use Lemma \ref{revenue-1} to bound the ratio between the welfare of \textsf{PBM-GSP} and the revenue of \textsf{PBM-VCG} with the Myerson reserve price. Here \textsf{PBM-VCG} is defined as the VCG mechanism in which bidding the expected value on any keyword is the dominant strategy and there is has no constraint on the total number of bid keywords. Finally, we bound the ratio between the revenue of \textsf{PBM-VCG} with the Myerson reserve price and the optimal welfare. For ease of reference, we use $SW_r(\mathbf{b})$ to denote the social welfare of \textsf{PBM-GSP}, and use $\mathcal{R}^{\mathcal{PBM-VCG}}_r(\mathbf{v})$ to denote the revenue of \textsf{PBM-VCG}, when reserve price vector $r$ is associated with these mechanisms.

The basic idea of the proof can be explained as follows.
\begin{eqnarray}
\mathbb{E}_{\mathbf{v},\mathbf{b}(\mathbf{v})}[\mathcal{R}_r(\mathbf{b}(\mathbf{v}))]\rightarrow \mathbb{E}_{\mathbf{v},\mathbf{b}(\mathbf{v})}[SW_r(\mathbf{b}(\mathbf{v}))]\rightarrow \mathbb{E}_{\mathbf{v}}[\mathcal{R}_r^{\mathcal{PBM-VCG}}(\mathbf{v})]\rightarrow \mathbb{E}_{\mathbf{v}}[SW(\mathcal{OPT}(\mathbf{v}))]\nonumber
\end{eqnarray}

\begin{lemma}\label{revenue-0}
If any advertiser's keyword valuation vector is i.i.d. drawn from an MHR distribution $T$ with bounded derivative $\eta$, the auction system is $\beta$-KL-expressive and $c$-expected-homogeneous, then for any Bayes-Nash equilibrium of \textsf{PBM-GSP} with the Myerson reserve price, the expected revenue is at least $\frac{1}{c e}$ of the expected welfare.
\end{lemma}
\begin{proof}
For any Bayes-Nash equilibrium $\mathbf{b}(\cdot)$, the expected social welfare of \textsf{PBM-GSP} with reserve price $r$ can be reformulated as follows:
{\small\begin{eqnarray}\label{Rev-00}
\mathbb{E}_{\mathbf{v},\mathbf{b}(\mathbf{v})}[SW_r(\mathbf{b}(\mathbf{v}))] &=& \mathbb{E}_{\mathbf{v},\mathbf{b}(\mathbf{v})}[\sum_s \int_{q \in N_G(s)}\pi_q(s)dP \sum_{k=1}^n w_k v_{\sigma_{s,\mathbf{b}(\mathbf{v})}(k)}^s \mathbb{I}[b_{\sigma_{s,\mathbf{b}(\mathbf{v})}(k)}^s \geq r_s]].\end{eqnarray}}
Since the auction system is $c$-expected-homogeneous, the following inequality holds with probability one,
{\small\begin{eqnarray}
c \mathbb{E}_{\mathbf{v}}[v_i^s] &=& \mathbb{E}_{\mathbf{v}}[\frac{\int_{q\in N_{G}(s)} c v_i^{q}\pi_{q}(s)dP}{\int_{q\in N_{G}(s)} \pi_{q}(s)dP}] = \frac{\int_{q\in N_{G}(s)} c \mathbb{E}_{\mathbf{v}}[v_i^{q}] \pi_{q}(s)dP}{\int_{q\in N_{G}(s)} \pi_{q}(s)dP} \geq \frac{\int_{q\in N_{G}(s)} v_i^{q} \pi_{q}(s)dP}{\int_{q\in N_{G}(s)} \pi_{q}(s)dP} = v_i^s.\nonumber
\end{eqnarray}}
Considering $T$ is an MHR distribution, we have $\mathbb{E}_{\mathbf{v}} [v_i^s] \leq e r_s (1 - T_s(r_s)) $, which yields,
{\small\begin{eqnarray} \label{Rev-01}
P(v_i^s \leq r_s c e)=1.
\end{eqnarray}}
Applying Eqn (\ref{Rev-01}) to Eqn(\ref{Rev-00}), we have
{\small\begin{eqnarray} \label{5.2app-1}
\mathbb{E}_{\mathbf{v},\mathbf{b}(\mathbf{v})}[SW_r(\mathbf{b}(\mathbf{v}))] &=& \mathbb{E}_{\mathbf{v},\mathbf{b}(\mathbf{v})}[\sum_s \int_{q \in N_G(s)}\pi_q(s)dP \sum_{k=1}^n w_k v_{\sigma_{s,\mathbf{b}(\mathbf{v})}(k)}^s \mathbb{I}[b_{\sigma_{s,\mathbf{b}(\mathbf{v})}(k)}^s \geq r_s]] \nonumber\\
&\leq& c e \mathbb{E}_{\mathbf{v},\mathbf{b}(\mathbf{v})}[\sum_s \int_{q \in N_G(s)}\pi_q(s)dP \sum_{k=1}^n w_k r_s \mathbb{I}[b_{\sigma_{s,\mathbf{b}(\mathbf{v})}(k)}^s \geq r_s]] \nonumber\\
&\leq& c e \mathbb{E}_{\mathbf{v},\mathbf{b}(\mathbf{v})}[\mathcal{R}_r(\mathbf{b}(\mathbf{v}))].\nonumber
\end{eqnarray}}
Then the lemma follows.
\end{proof}
\begin{lemma}\label{revenue-1}
If any advertiser's keyword valuation vector is i.i.d. drawn from an MHR distribution $T$ with bounded derivative $\eta$, the auction system is $\beta$-KL-expressive and $c$-expected-homogeneous, then the expected welfare in any Bayes-Nash equilibrium of \textsf{PBM-GSP} with the Myerson reserve price is at least $\frac{\beta}{1+\beta} \frac{1}{2\eta}$ of the expected revenue of \textsf{PBM-VCG} with the Myerson reserve price.
\end{lemma}
\begin{proof}
The proof technique we use here can be regarded as a variation of that in \cite{leme2010pure,roughgarden2009intrinsic}. Without loss of generality, we assume that if any advertiser bids on any keyword $s$, his/her bid price is larger than $r_s$.

Given the Myerson reserve prices on any keyword, denote $S_i$ as the set of keywords that advertiser $i$ can win a slot when all advertisers truthfully bid. We consider advertiser $i$ with a specific randomized strategy $b'_i(\cdot)$ that randomly chooses $\kappa$ keywords from $S_i$ and bid ${b'}_i^{s}(v_i) = \frac{v_i^s+r_s}{2}$ if keyword $s$ is sampled by the strategy.

Given any pure bid profile $\mathbf{b}$, any value profile $\mathbf{v}$ and any position $k$, given any sampled keyword $s$ by the strategy $b'_i(v_i)$, if advertiser $i$ changes his/her bid to $b'_i(v_i)$ and the position he/she gets is not lower than position $k$, his/her per-click utility must be larger than $v_i^{s} -{b'}_i^{s}(v_i)$; otherwise, the value of the advertiser ranked at the position $k$ must be larger than ${b'}_i^{s}(v_i)$. Then we have
{\small\begin{eqnarray}
u^s_i({b'}^{s}_i(v_i),\mathbf{b}^s_{-i}) &\geq&   \mathbb{I}[{b'}_i^{s}(v_i) > b_{\sigma_{s,\mathbf{b}}(k)}^{s}] w_k \int_{q \in N_G(s)} \pi_q(s) (v_i^{s} - {b'}_i^{s}(v_i)) dP \nonumber\\
&-&\mathbb{I}[{b'}_i^{s}(v_i) \leq b_{\sigma_{s,\mathbf{b}}(k)}^{s}] w_{k}\int_{q \in N_G(s)} \pi_q(s) (v_{\sigma_{s,\mathbf{b}}(k)}^s - {b'}_i^{s}(v_i)) dP.\nonumber\\
&\geq& w_k \int_{q \in N_G(s)} \pi_q(s) (v_i^{s} - {b'}_i^{s}(v_i)) dP- w_k \int_{q \in N_G(s)} \pi_q(s) v_{\sigma_{s,\mathbf{b}}(k)}^s dP.
\end{eqnarray}}
Since distribution $T$ is an MHR distribution with bounded derivative $\eta$, we have $v_i^s - {b'}_i^s(v_i) = \frac{v_i^s-r_s}{2} \geq \frac{\phi_s(v_i^s)}{2\eta}$, which yields,
{\small\begin{eqnarray}
u^s_i({b'}^{s}_i(v_i),\mathbf{b}^s_{-i})\geq w_{k} \int_{q \in N_G(s)} \pi_q(s) \frac{\phi_s(v_i^s)}{2\eta} dP -w_{k} \int_{q \in N_G(s)} \pi_q(s) v_{\sigma_{s,\mathbf{b}}(k)}^{s} dP. \label{revenue_121}
\end{eqnarray}}
Let $k$ be the position of advertiser $i$ on keyword $s$ when all advertiser truthfully bid, i.e., $k=\sigma_{s, \mathfrak{v}}^{-1}(i)$. By summing over all advertisers and all keywords, and taking expectation over the valuation profile $\mathbf{v}$, bidding strategy $\mathbf{b}(\mathbf{v})$ and the strategy $\mathbf{b}'(\cdot)$, we have

{\small\begin{eqnarray}
\mathbb{E}_{\mathbf{v},\mathbf{b}(\mathbf{v}),\mathbf{b}'(\mathbf{v})}[\sum_{i=1}^n u_i({b'}_i(v_i),\mathbf{b}_{-i}(\mathbf{v}_{-i})]
&\geq& \beta \mathbb{E}_{\mathbf{v},\mathbf{b}(\mathbf{v})}[\sum_{i=1}^n\sum_{s \in S_i} u^s_i({b'}^{s}_i(v_i),\mathbf{b}^s_{-i}(\mathbf{v}_{-i})] \nonumber\\
&\geq& \beta \mathbb{E}_{\mathbf{v},\mathbf{b}(\mathbf{v})} [\sum_{i=1}^n\sum_{s \in S_i}w_{\sigma_{s, \mathfrak{v}}^{-1}(i)} \int_{q \in N_G(s)} \pi_q(s) \frac{\phi_s(v_i^s)}{2\eta} dP] \nonumber \\
&-& \beta\mathbb{E}_{\mathbf{v},\mathbf{b}(\mathbf{v})} [\sum_{i=1}^n\sum_{s \in S_i} w_{\sigma_{s, \mathfrak{v}}^{-1}(i)} \int_{q \in N_G(s)} \pi_q(s) v_{\sigma_{s,\mathbf{b}}(\sigma_{s, \mathfrak{v}}^{-1}(i))}^{s}dP] \nonumber \\
&\geq& \beta \mathbb{E}_{\mathbf{v},\mathbf{b}(\mathbf{v})} [\sum_{i=1}^n\sum_{s \in S}w_{\sigma_{s, \mathfrak{v}}^{-1}(i)} \int_{q \in N_G(s)} \pi_q(s) \frac{\phi_s(v_i^s)\mathbb{I}[\phi_s(v_i^s) > 0]}{2\eta} dP] \nonumber \\
&-& \beta\mathbb{E}_{\mathbf{v},\mathbf{b}(\mathbf{v})} [\sum_{k=1}^n\sum_{s \in S_i} w_{k} \int_{q \in N_G(s)} \pi_q(s) v_{\sigma_{s, \mathbf{b}}(k)}^{s}dP] \nonumber \\
&\geq&  \frac{\beta}{2\eta}\mathbb{E}_{\mathbf{v}}[\mathcal{R}_r^{\mathcal{PBM-VCG}}(\mathbf{v})] - \beta \mathbb{E}_{\mathbf{v},\mathbf{b}(\mathbf{v})}[SW_r(\mathbf{b}(\mathbf{v}))]\nonumber
\end{eqnarray}}

Similar to the proof of Theorem \ref{welfare1}, we have for any Bayes-Nash equilibrium $\mathbf{b}(\cdot)$,
{\small\begin{eqnarray}\label{5.2app-2}
\mathbb{E}_{\mathbf{v},\mathbf{b}(\mathbf{v})}[SW_r(\mathbf{b}(\mathbf{v}))] &\geq& \mathbb{E}_{\mathbf{v},\mathbf{b}(\mathbf{v})}[u_i(\mathbf{b}(\mathbf{v}))] \geq
\mathbb{E}_{\mathbf{v},\mathbf{b}(\mathbf{v}),\mathbf{b}'(\mathbf{v})}[\sum_i u_i({b'}_i(v_i), \mathbf{b}_{-i}(\mathbf{v_{-i}}))] \nonumber\\
&\geq& \frac{\beta}{2\eta} \mathbb{E}_{\mathbf{v}}[\mathcal{R}_r^{\mathcal{PBM-VCG}}(\mathbf{v})] - \beta \mathbb{E}_{\mathbf{v},\mathbf{b}(\mathbf{v})}[SW_r(\mathbf{b}(\mathbf{v}))].
\end{eqnarray}}
Then the lemma follows.
\end{proof}
Now we give the overall proof of Theorem \ref{revenue-2}.
\begin{proof}[of Theorem 5.2]
Combining Lemma \ref{revenue-0} and Lemma \ref{revenue-1}, we have for any Bayes-Nash equilibrium $\mathbf{b}(\cdot)$ of \textsf{PBM-GSP} with the Myerson reserve price, the following holds:
{\small\begin{eqnarray}\label{5.2app-1}
\mathbb{E}_{\mathbf{v},\mathbf{b}(\mathbf{v})}[\mathcal{R}_r(\mathbf{b}(\mathbf{v}))] \geq \frac{\beta}{2ce(\beta + 1)\eta}\mathbb{E}_{\mathbf{v}}[\mathcal{R}_r^{\mathcal{PBM-VCG}}(\mathbf{v})].
\end{eqnarray}}
Denote $\mathcal{R}_r^{\mathcal{PBM-VCG}}(\mathbf{v},s)$ as the revenue obtained from keyword $s$ of the VCG mechanism with reserve price vector $r$, and denote $SW (\mathfrak{v},s)$ as the welfare obtained from keyword $s$ of with the VCG mechanism. For any given keyword $s$, denote $z_i^s(x)=\mathbb{E}_{\mathbf{v}} [w_{\sigma_{s,\mathfrak{v}}^{-1}(i)}|v_i^s=x]$. Then it is easy to see that $z_i^s(x)$ is an increasing function since advertiser will obtain a better postion in a welfare-maximizing allocation with a larger valuation. Denote $\phi_s^+(x)=\phi_s(x) \mathbb{I}[\phi_s(x) > 0]$. According to Lemma 5.1 in \cite{2013ratio}, we have
{\small\begin{eqnarray}
\frac{\mathbb{E}_{\mathbf{v}}[\phi_s^+(v^s_i)w_{\sigma_{s,\mathfrak{v}}^{-1}(i)}]}{\mathbb{E}_{\mathbf{v}}[v^s_iw_{\sigma_{s,\mathfrak{v}}^{-1}(i)}]}= \frac{\mathbb{E}_{\mathbf{v}}[\phi_s ^+(v^s_i)\mathbb{E}_{\mathbf{v}}[w_{\sigma_{s,\mathfrak{v}}^{-1}(i)}|v_i^s]]}{\mathbb{E}_{\mathbf{v}}[v^s_i\mathbb{E}_{\mathbf{v}}[w_{\sigma_{s,\mathfrak{v}}^{-1}(i)}|v_i^s]]} =\frac{\mathbb{E}_{\mathbf{v}}[\phi_s^+(v^s_i)z_i^s(v^s_i)]}{\mathbb{E}_{\mathbf{v}}[v^s_iz_i^s(v^s_i)]}\geq \frac{\mathbb{E}_{\mathbf{v}}[\phi_s ^+(v^s_i)]}{\mathbb{E}_{\mathbf{v}}[v^s_i]}\geq \frac{1}{e}
\end{eqnarray}}
That is, $\mathbb{E}_{\mathbf{v}}[\phi_s^+(v^s_i)w_{\sigma_{s,\mathfrak{v}}^{-1}(i)}] \geq \frac{1}{e}\mathbb{E}[v^s_iw_{\sigma_{s, \mathfrak{v}}^{-1}(i)}]$.
By summing over all advertisers and summing over all keywords, we have $\mathbb{E}_{\mathbf{v}}[\mathcal{R}_r^{\mathcal{PBM-VCG}}(\mathbf{v})] \geq \frac{1}{e} \mathbb{E}_{\mathbf{v}}[SW(\mathfrak{v})]$.
According to the second step of the proof for Theorem \ref{welfare1}, we have $\mathbb{E}_{\mathbf{v}}[SW(\mathfrak{v})] \geq \frac{1}{c} \mathbb{E}_{\mathbf{v}}[SW(\mathcal{OPT}(\mathbf{v}))]$. By combining with (\ref{5.2app-1}), we prove the theorem.
\end{proof}

Theorem \ref{revenue-2} can be further improved if we are only interested in the single-slot case.
\begin{theorem}\label{revenue-21}
If any advertiser's keyword valuation profile is i.i.d. drawn from an MHR distribution $T$ with bounded derivative $\eta$, the auction system is $\beta$-KL-expressive and $c$-homogeneous, and there is only one ad slot to sell, then the revenue obtained by the \textsf{PBM-GSP} mechanism with the Myerson reserve price is at least $\frac{\beta}{1+\beta} \frac{1}{\eta (c e)^2}$ of the optimal social welfare.
\end{theorem}
\begin{proof}
We only need to slightly modify the proof of Lemma \ref{revenue-1} and get

$(\beta + 1)\mathbb{E}[SW_r(\mathbf{b}(\mathbf{v}))] \geq \frac{\beta}{\eta} \mathbb{E}[\mathcal{R}_r^\mathcal{PBM-VCG}(\mathbf{v})]$.
in the single-slot case.

Denote $S_i$ as the set of keywords that advertiser $i$ can win when all advertisers truthfully bid, we consider advertiser $i$ and a specific randomized strategy $b'_i(\cdot)$ that randomly chooses $\kappa$ keyword among $S_i$ and bid $v_i^s$ if keyword $s$ is sampled.

For any pure strategy $\mathbf{b}$ and any value profile $\mathbf{v}$, it can be proven that for any given sampled keyword $s$ by the strategy $b'_i(v^s_i)$,
{\small\begin{eqnarray}
u_i^s({b'}_i^s(v_i),\mathbf{b}^s_{-i}) &\geq&\mathbb{I}[b^s_{\sigma_{s,\mathbf{b}}(1)}\geq r_s] \int_{q\in N_G(s)} \pi_q(s) (v_i^{s} - b^s _{\sigma_{s,\mathbf{b}}(1)})dP \\
&+&\mathbb{I}[ b^s _{\sigma_{s,\mathbf{b}}(1)}< r_s] \int_{q\in N_G(s)} \pi_q(s) (v_i^{s} - r_s ) dP  \\
&\geq&  \int_{q\in N_G(s)} \pi_q(s) (v_i^{s} - r_s - v^s _{\sigma_{s,\mathbf{b}}(1)}) dP.
\end{eqnarray}}
Since distribution $T$ is an MHR distribution with bounded derivative $\eta$, we have
{\small\begin{eqnarray}
u_i^s({b'}^s_i(v_i),\mathbf{b}^s_{-i}) &\geq& \int_{q\in N_G(s)} \pi_q(s) (\frac{\phi_s(v_i^s)}{\eta} - v^s _{\sigma_{s,\mathbf{b}}(1)} ) dP.
\end{eqnarray}}

The proofs follows by proceeding other parts of the proof of Lemma \ref{revenue-1}
\end{proof}
\section{Related Works}
In this section, for the sake of completeness, we give an overview of the related works to the paper. Overall, the related works can be categorized into three groups.

First, there have been a rich literature of theoretical analysis on GSP auctions. For example, \cite{leme2010pure,lucier2011gsp,caragiannis2011efficiency} analyze the PoA when bidders are conservative, they show that the pure PoA is at most $1.282$, mixed-strategy PoA is at most $2.310$ and Bayes-Nash PoA is at most $2.927$ in GSP auction. Some other works analyze the revenue of GSP. In \cite{edelman2007strategic,varian2007position}, it is shown that GSP's revenue is at least as good as VCG in envy-free equilibrium. In \cite{lucier2012revenue,caragiannis2012revenue}, the revenue of GSP over all Bayes-Nash equilibrium is studied and a ratio bound ($4.72$ for regular distribution and 3.46 for MHR distribution) is given between the optimal auction and GSP with a proper reserve price in the Bayesian setting.

Second, there are a few works that pay attention to the broad-match mechanism, and in particular the \textsf{SBM} mechanism. There have been several pieces of work that study the optimization problems regarding \textsf{SBM}. For example, in \cite{feldman2007budget}, the budget optimization problem is considered and a $(1-1/e)$-approximation algorithm is developed. In \cite{even2009bid}, it is shown that the bid optimization problem regarding \textsf{SBM} is NP-Hard and is inapproximable with any reasonable approximation factor unless $P=NP$. Some other works perform PoA analysis on the \textsf{SBM-GSP} mechanism. In \cite{dhangwatnotai2011multi}, by assuming advertisers to play undominated strategies, the authors develop an almost-tight bound for the pure PoA of \textsf{SBM-GSP} in the single-slot case.

Third, the design principle of our proposed \textsf{PBM} mechanism is also related to the probabilistic single-item auctions with mixed signals \cite{bro2012send,emek2012signaling, 2013CPS}. A \emph{probabilistic single-item auction} is defined as follows. The auctioneer wishes to sell the items drawn from an item set $\mathcal{Q}$ according to a known distribution $P$ to $n$ bidders. Each bidder $i$ has a valuation of $v_i(q)$ on item $q$, but he/she cannot directly observe the item before he/she bids. At each time, the auctioneer draws an item and broadcasts a signal to the bidders according to a signaling scheme defined at the very beginning of the auction. The signaling scheme can be probabilistic, and can be strategically designed by the auctioneer. After receiving the signals, the bidders submit their bids on the signals, and the item will be allocated and charged to one of the bidders by using the second price auction. If we define the item set as the query space, define the signals as the keywords, and define the signaling scheme based on the query-keyword bipartite graph and the matching probability distribution $\pi_q$, then the above problem will become very similar to our \textsf{PBM} problem. However, we would like to point out three critical differences between them. (1) In \cite{bro2012send,emek2012signaling}, the whole signaling scheme can be strategically chosen by the auctioneer, but in our setting, the signals (keywords) that can be broadcasted given an item (query) is restricted according to the bipartite graph. The strategy of the auctioneer only lies in the design of the matching probability distribution. (2) In \cite{2013CPS}, the author analyzed the situation that each signal has one winner, like the single slot in our case and assumed that all participants will truthfully report their value on it. With this, the author try to optimize a matching probability distribution, which is restricted according to the bipartite graph, and it has constant approximation to the optimal social welfare and revenue bound. (3) In our problem, each bidder is allowed to bid on only up to $\kappa$ signals (keywords) and, for each keyword, auctioneer will conduct a GSP on it. As a result, the truth telling will not be (always not) a dominant strategy any longer. These additional constraints will increase the difficulty of the problem and the techniques developed in \cite{bro2012send,emek2012signaling,2013CPS} need to be enhanced or extended to fit into our setting.

\section{Conclusion}
In this paper, we propose a probabilistic broad-match mechanism for sponsored search. We show that this new mechanism has better theoretical guarantees than the currently used broad-match mechanism in terms of both social welfare and search engine revenue. We have summarized our key results in Table \ref{resultsummary1}, Table \ref{resultsummary2} for ease of reference.

For future work, we plan to work on the following topics. First, we will work on the optimization of the matching probability in the proposed mechanism so as to maximize the social welfare or revenue. Second, we will investigate if there is a tighter bound for our results. Third, we will perform more theoretical analysis on the currently used broad-match mechanism, which is far from complete in the literature.
\begin{table}
\centering
\tbl{Summary of Social Welfare Analysis for the \textsf{PBM-GSP} Mechanism}{
\footnotesize%

\begin{tabular}{|c|c|c|c|}
    \hline
  \multicolumn{2}{|c|}{} & Multi-slot & Single-slot \\\hline
  \multirow{2}{*}{Social welfare} & Bayesian &  $\frac{(e-1)\beta}{ec(\beta + 1)}SW(OPT)$ & $\frac{\beta}{c(\beta + 1)}SW(OPT)$ \\\cline{2-4}
& Full-information & $\frac{\beta}{c(\beta + 1)}SW(OPT)$ & $\frac{\beta}{c} SW(OPT)$\\
  \hline
\end{tabular}}
\label{resultsummary}
\label{resultsummary1}
\end{table}
\begin{table}
\centering
\tbl{Summary of Revenue Analysis for the \textsf{PBM-GSP} Mechanism}{
\footnotesize%

\begin{tabular}{|c|c|c|c|}
    \hline
\multicolumn{2}{|c|}{} &Multi-slot& Single-slot\\\hline
  Revenue & Bayesian &$\frac{\beta}{(1+\beta)(ce)^2 2 \eta}SW(OPT)$  & $\frac{\beta}{(1+\beta)(ce)^2 \eta}SW(OPT)$ \\
  \hline
\end{tabular}}
\label{resultsummary2}
\end{table}

\bibliographystyle{acmsmall}
\small
\bibliography{PBM_GSP}


\elecappendix
\normalsize
\medskip

\section{APPENDEX: Comparison Between \textsf{PBM-GSP} mechanism and \textsf{SBM-GSP} mechanism}

\subsection{Theoretical Comparison Between Two Expressiveness Measures}

There are both differences and connections between KL-expressiveness and QL-expressiveness.
First, KL-expressiveness is focused on the coverage of positive keywords while QL-expressiveness is focused on the coverage of positive queries. Second, given the query-keyword bipartite graph $G$, each advertiser's positive query set $Q_i$, and $\kappa$, the value of KL-expressiveness can be computed in linear time, while determining the value of QL-expressiveness is NP-Hard in general (since the set cover problem is its sub routine). Given the same query-keyword bipartite graph, the KL-expressiveness and QL-expressiveness are actually tightly coupled due to the mutual bounds given in the following proposition.

\begin{proposition}\label{proposition123}
If the maximum degree of the query-keyword bipartite graph $G$ is bounded by $\gamma$, and the auction system is $\alpha$-QL-expressive and $\beta$-KL-expressive, then $\frac{\alpha}{\gamma^2}\leq\beta\leq \gamma\alpha$.
\end{proposition}
\begin{proof}
First we prove $\beta\leq \gamma\alpha$. It is clear that any $\kappa$ queries can be covered by $\kappa$ keywords, thus we have $\alpha\geq\min_i\frac{\kappa}{|Q_i|}$. Considering that $|Q_i|$ should be smaller than $\gamma$ times $|\{ s: N_G(s)\cap Q_i\neq\emptyset\}|$, which is the number of positive keywords for advertiser $i$, we have $\alpha\geq\min_i\frac{\kappa}{|Q_i|}\geq\min_i\frac{\kappa}{\gamma|\{ s: N_G(s)\cap Q_i\neq\emptyset\}|}=\frac{\beta}{\gamma}$. Next we prove $\frac{\alpha}{\gamma^2}\leq\beta$. By definition, $\beta=\min_i\frac{\kappa}{|\{ s: N_G(s)\cap Q_i\neq\emptyset\}|}$. Considering $\kappa$ keywords can cover at most $\kappa \gamma$ queries, we have $\beta=\min_i\frac{\kappa \gamma}{|\{ s: N_G(s)\cap Q_i\neq\emptyset\}|\gamma} \geq \min_i\frac{\alpha |Q_i|}{|\{ s: N_G(s)\cap Q_i\neq\emptyset\}|\gamma}\geq \frac{\alpha}{\gamma^2}$.
\end{proof}
The bounds given in Proposition \ref{proposition123} depend on $\gamma$, the degree of the query-keyword bipartite graph. When $\gamma$ is large, the bounds become useless. In this case, it would be more meaningful to directly compare the values of QL-expressiveness and KL-expressiveness. This is exactly what we do in the next subsection.

\subsection{Empirical Comparison Between Two Expressiveness Measures}

We base our empirical study on the log data obtained from a commercial search engine, which contains the query-keyword bipartite graph and advertiser's bid keywords in a one-week time frame. Please note that even with this real data, it is still highly non-trivial to conduct empirical study due to the following reasons. (1) The computation of the QL-expressiveness is NP-Hard in general (since the set cover problem is its sub routine), which prevents us from doing experiments on very large data. (2) Both definitions of expressiveness require knowledge about the positive queries for an advertiser, which is unknown in practice (we only know the bid prices on the keywords). To tackle these challenges, we have designed our experiments as follows.

First, we restrict our empirical study to small micro markets. A micro market refers to the queries and ads (and also their bid keywords) that concentrate on a specific product. A micro market can be roughly considered as a closed system, and the expressiveness in different micro markets can be treated separately. For example, the queries and ads about ``insurance" form a micro market. Other examples of micro market include ``travel", ``hotel", and ``car". In this work, we employ a simple and straightforward method to identify micro markets, and define the size of a micro market using the number of keywords in it.\footnote{We simply use \emph{term sharing} as the rule to define micro markets. That is, if a set of keywords and queries contain the same term A (e.g., insurance), we will consider them to belong to the same micro market. We take this simple approach because we did not find previous works that can be used to fulfill the task. We believe different ways of defining micro markets will not significantly affect our experimental results; however, we are willing to adopt more advanced approaches when they are available in the future. Furthermore, we use the number of keywords to define the size of a micro market because it is the most critical factor in the computation of QL-expressiveness.} To ensure that the computation of the QL-expressiveness is feasible, we randomly sample 1000 micro markets whose sizes are smaller than 20, and use them for our experimental study.

Second, we simulate the value on a query using a similarity based approach. Specifically, we assume that if a query is similar enough (measured by a threshold) to the keyword that an advertiser bids on, it will be a positive query. In our experiment, we compute the similarity between query $q$ and keyword $s$ based on the Levenshtein distance $d(q,s)$, which is popularly used in information retrieval and usually referred to as the edit distance. Informally speaking, the Levenshtein distance equals the minimal number of single-character edits required to change query $q$ to keyword $s$. Based on the Levenshtein distance, we define the similarity function as $\text{Sim}(q,s) = 1 - \frac{d(q,s)}{\text{maxlength}(q,s)}$.

Then if we observe that advertiser $i$ bids on keywords $\{s_1,s_2,\cdots,s_m\}$ in the historical auction logs, we define the set of positive queries as follows,
{\small\begin{equation}
Q_i(\theta)= \{q\in\mathcal{Q}: \exists s\in\{s_1,s_2,\cdots,s_m\}, \text{Sim}(q,s)>\theta \}.
\end{equation}}
When $\kappa$ and $\theta$ are given, the values of both KL-expressiveness $\beta$ and QL-expressiveness $\alpha$ for each micro market can be computed. In our experiments, we change $\theta$ from 0.9 to 0 and change $\kappa$ from 1 to the size of the micro market, so as to generate a large number of ($\alpha, \beta$) pairs. We conduct some statistical significance test on these data points, and find that $\beta$ is larger than $\alpha/3$ with a p-value$=0.01$. This gives a very accurate quantitative relationship between the two notions of expressiveness based on real data.
\begin{figure}
\centering \label{discount}
\epsfig{file=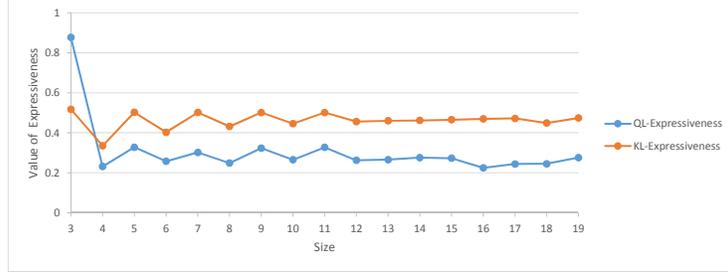,width=0.7\textwidth}
\caption{Expressiveness values w.r.t. size of micro market}
\end{figure}

\begin{table}%
\tbl{QL-expressiveness on real data\label{tab:one}}%
{\footnotesize
\begin{tabular}{|l|l|l|l|l|l|l|l|l|l|l|l|}
\hline
\backslashbox{$\theta$}{$\kappa/$size}& 0-10\% & 10-20\% & 20-30\% & 30-40\% & 40-50\% & 50-60\% & 60-70\% & 70-80\% & 80-90\% & 90-100\% \\\hline
0.9 & 0.518	& 0.743	& 0.878	& 0.952	& 0.980	& 0.990	& 0.996	& 0.999	& 0.999	& 1 \\\hline
0.8 & 0.432	& 0.670	& 0.826	& 0.923	& 0.964	& 0.982	& 0.993	& 0.998	& 0.999	& 1 \\\hline
0.7 & 0.247	& 0.492	& 0.687	& 0.834	& 0.917	& 0.960	& 0.983	& 0.990	& 0.996	& 0.999 \\\hline
0.6 & 0.123	& 0.290	& 0.485	& 0.658	& 0.803	& 0.875	& 0.942	& 0.972	& 0.991	& 0.999 \\\hline
0.5 & 0.080	& 0.179	& 0.314	& 0.481	& 0.653	& 0.727	& 0.838	& 0.915	& 0.969	& 0.998 \\\hline
0.4 & 0.055	& 0.115	& 0.203	& 0.326	& 0.509	& 0.571	& 0.710	& 0.808	& 0.920	& 0.994 \\\hline
0.3 & 0.044	& 0.094	& 0.157	& 0.243	& 0.409	& 0.428	& 0.586	& 0.691	& 0.842	& 0.981 \\\hline
0.2 & 0.039	& 0.084	& 0.141	& 0.207	& 0.277	& 0.360	& 0.501	& 0.617	& 0.770	& 0.973 \\\hline
0.1 & 0.038	& 0.083	& 0.138	& 0.201	& 0.260	& 0.351	& 0.485	& 0.590	& 0.756	& 0.968 \\\hline
0 & 0.038	& 0.082	& 0.138	& 0.200	& 0.259	& 0.349	& 0.484	& 0.587	& 0.753	& 0.965 \\\hline
\end{tabular}}
\label{QLexp}
\end{table}%

\begin{table}%
\tbl{KL-expressiveness on real data\label{tab:one}}{\footnotesize%
\begin{tabular}{|l|l|l|l|l|l|l|l|l|l|l|l|}
\hline
\backslashbox{$\theta$}{$\kappa/$size}& 0-10\% & 10-20\% & 20-30\% & 30-40\% & 40-50\% & 50-60\% & 60-70\% & 70-80\% & 80-90\% & 90-100\% \\\hline
0.9 & 0.222	& 0.431	& 0.633	& 0.765	& 0.858	& 0.915	& 0.950	& 0.974	& 0.991	&0.999\\\hline
0.8 & 0.202	& 0.404	& 0.607	& 0.745	& 0.840	& 0.902	& 0.941	& 0.967	& 0.988	& 0.999\\\hline
0.7 & 0.167	& 0.342	& 0.536 & 0.682	& 0.794	& 0.866	& 0.913	& 0.953	& 0.980	& 0.998\\\hline
0.6 & 0.126	& 0.261	& 0.425	& 0.570	& 0.705	& 0.796	& 0.862	& 0.922	& 0.963	& 0.996\\\hline
0.5 & 0.101	& 0.211	& 0.345	& 0.473	& 0.613	& 0.720	& 0.800	& 0.878	& 0.941 & 0.994\\\hline
0.4 & 0.085	& 0.177	& 0.292	& 0.403	& 0.535	& 0.643	& 0.731	& 0.829	& 0.912	& 0.990\\\hline
0.3 & 0.078	& 0.162	& 0.266	& 0.369	& 0.491	& 0.592	& 0.681 & 0.781	& 0.881	& 0.986\\\hline
0.2 & 0.075	& 0.156	& 0.257	& 0.357	& 0.474	& 0.571	& 0.658	& 0.759	& 0.860	& 0.982\\\hline
0.1 & 0.074	& 0.155	& 0.255	& 0.355	& 0.470	& 0.567	& 0.655	& 0.755	& 0.855	& 0.980\\\hline
0 & 0.074	& 0.155	& 0.255	& 0.354	& 0.470	& 0.567 & 0.655	& 0.755	& 0.855	& 0.980\\\hline
\end{tabular}}
\label{KLexp}
\end{table}

To get a more friendly view of the data points, we create Tables \ref{QLexp} and \ref{KLexp}. Since the sizes of different micro markets can vary largely, we normalize $\kappa$ with the size of the micro market and quantify the values into ten buckets. For each bucket we calculate the average $\alpha$ and $\beta$ values as listed in the tables. From the tables, we can see that in each bucket, with the increasing number of positive queries, both $\alpha$ and $\beta$ become smaller. On the other hand, if the number of positive queries is fixed, when the normalized $\kappa$ grows, both $\alpha$ and $\beta$ become larger.

Due to the computational complexity, we only use the micro markets whose sizes are smaller than 20 in our experiments. One may doubt whether our conclusion can be generalized to larger micro markets. Our answer is positive. This is because the values of $\alpha$ and $\beta$ have become stable when the size of the micro markets is around 10. For each element in Tables \ref{QLexp} and \ref{KLexp}, we can plot a figure showing the comparison between $\alpha$ and $\beta$ with respect to the size of the micro market. We find that the figures for all the elements demonstrate the same trend. Here we give one example in Figure 1 (corresponding to $\theta=0.2$ and $\kappa=0.4\cdot size$). From the figure, we can see that $\beta$ approaches 0.45 and $\alpha$ stabilizes to around 0.25 very quickly. \footnote{For completeness, we put all the figures at \url{http://research.microsoft.com/en-us/people/tyliu/ec2013-appendix.zip}} Therefore we can expect that the comparison between $\alpha$ and $\beta$ has stabilized and the conclusion will not change by much for larger micro markets.

\end{document}